\documentclass[sigconf]{acmart}
\usepackage{booktabs} 
\usepackage{graphicx,array}
\usepackage{balance}
\usepackage{color,soul,xspace}
\usepackage{times}
\usepackage{hyperref}

\usepackage{algorithmic}
\usepackage[ruled,vlined,linesnumbered]{algorithm2e}
\usepackage{comment}
\usepackage{epsfig}
\usepackage{subfig}
\usepackage{mathrsfs, amsmath,mdwlist,enumitem}
\DeclareMathOperator*{\argmax}{arg\,max}
\newtheorem{example}{\textbf{Example}}
\newtheorem{thm}{Theorem}
\newtheorem{observation}{Observation}

\newtheorem{definition}{Definition}
\newtheorem{problem}{Problem}

\setlist{nolistsep,leftmargin=*}

\setcopyright{rightsretained}

\acmDOI{10.475/123_4}

\acmISBN{123-4567-24-567/08/06}

\acmConference[WOODSTOCK'97]{ACM Woodstock conference}{July 1997}{Texas USA}
\acmYear{1997}
\copyrightyear{2016}

\acmArticle{4}
\acmPrice{15.00}

\editor{Jennifer B. Sartor}
\editor{Theo D'Hondt}
\editor{Wolfgang De Meuter}

\begin{document}
\title{Influence Minimization Under Budget and Matroid Constraints: Extended Version}
\author{ Sourav Medya}
\affiliation{%
  \institution{University of California, Santa Barbara}
}
\email{medya@cs.ucsb.edu}

\author{Arlei Silva}
\affiliation{%
  \institution{University of California, Santa Barbara}
   \state{}
}
\email{arlei@cs.ucsb.edu}

\author{Ambuj Singh}
\affiliation{%
\institution{University of California, Santa Barbara}
}
\email{ambuj@cs.ucsb.edu}

\begin{abstract}


Recently, online social networks have become major battlegrounds for political campaigns, viral marketing, and the dissemination of news. As a consequence, ''bad actors'' are increasingly exploiting these platforms, becoming a key challenge for their administrators, businesses and the society in general. The spread of fake news is a classical example of the abuse of social networks by these actors. While some have advocated for stricter policies to control the spread of misinformation in social networks, this often happens in detriment of their democratic and organic structure. In this paper we study how to limit the influence of a target set of users in a network via the removal of a few edges. The idea is to control the diffusion processes while minimizing the amount of disturbance in the network structure.

We formulate the influence limitation problem in a data-driven fashion, by taking into account past propagation traces. Moreover, we consider two types of constraints over the set of edge removals, a budget constraint and also a, more general, set of matroid constraints. These problems lead to interesting challenges in terms of algorithm design. For instance, we are able to show that influence limitation is APX-hard and propose deterministic and probabilistic approximation algorithms for the budgeted and matroid version of the problem, respectively. Our experiments show that the proposed solutions outperform the baselines by up to 40\%.

\end{abstract}

%
%

\keywords{Influence Minimization, Network Modification, Matroid}

\maketitle

\section{Introduction}


Online social networks, such as Facebook and Twitter, were popularized mostly as platforms for sharing entertaining content and maintaining friendship and family ties. However, they have been quickly transformed into major battlegrounds for political campaigns, viral marketing, and the dissemination of news. With this shift, the increase in the number of ``bad actors'', such as tyrannical governments, spammers, hackers, bots, and bullies exploiting these platforms has become a key challenge not only for their administrators but for businesses and society in general. 

A classical example of the abuse of social networks is the spread of fake news. 
As a concrete example, Starbucks\footnote{\url{http://uk.businessinsider.com/fake-news-starbucks-free-coffee-to-undocumented-immigrants-2017-8}} recently was the victim of a hoax claiming that it would give free coffee to undocumented immigrants~\cite{Tschiatschek2018}.
Earlier, Twitter had a vast number of threat reports with inaccurate locations where riots would take place across the UK. People were terrified as false reports of riots in their local neighborhoods broke on social media \cite{bogunovic2012robust}. A fundamental question is: \textit{How can one (e.g., Starbucks, governments) limit the spread of misinformation in social networks? } 

A questionable approach to control the diffusion of misinformation in social platforms is via stricter laws and regulations by governments. This control often happens in detriment of the democratic and organic structure that are central to these platforms. Instead, a more sensible approach is to limit the impact of bad actors in the network while minimizing the disruption of its structure. In this paper we formalize this general problem as \textit{the influence minimization problem}. In particular, we focus on a setting where the network is modified via the removal of a few edges. These modifications might be implemented by social network administrators or induced by other organizations or governments via advertising campaigns.


The problem of controlling influence spread via structural changes in a network has attracted recent interest from the research community \cite{Tong2012GML,kuhlman2013blocking, Khalil2014}. However, existing work assumes that diffusion follows classical models from the literature---e.g., Independent Cascade, Linear Threshold, and Susceptible Infected Recovered. These models are hard to validate at large-scale while also requiring computationally-intensive simulations in order to evaluate the effect of modifications. Instead, we propose a data-driven approach for influence minimization based on propagation traces \cite{goyal2011}. More specifically, our modifications are based on historical data, which makes our solutions less dependent on a particular diffusion model. 

Another important aspect of the influence minimization problem considered in this work is the type of constraint imposed on the amount of modification allowed in the network. 
The influence limitation (minimization) problems are often studied under budget constraints \cite{Tong2012GML,kuhlman2013blocking}, where a fixed number of edges can be blocked in the network. One of the main advantages of this type of formulation is that the associated objective function is often monotone and submodular, enabling the design of a simple greedy algorithm that achieves good approximation guarantees \cite{Khalil2014}. On the other hand, budget constraints have undesired effects in many settings. For instance, they might disconnect or disproportionately affect particular sub-networks. Besides disturbing the network structure, such effects are in conflict with important modern issues, such as algorithmic fairness \cite{aziz2018knowledge}. We address this issue by studying the influence limitation problem not only under a budget constraint but also under matroid constraints \cite{nemhauser1978,chekuri2004maximum}. Our formulations provide an interesting comparison between these constraints and showcase the expressive power of matroids for problems defined over networks.

The main goal of this paper is to show how the formalization of the influence limitation problem under budget and matroid constraints leads to interesting challenges in terms of algorithm design. Different from the budget version, for which we propose a simple greedy algorithm, the matroid version requires a more sophisticated solution via continuous relaxation and rounding. Yet, we provide a theoretical analysis of the performance of both algorithms that is supported by the fact that the objective function of the influence limitation problem is submodular. Moreover, we provide strong inapproximability results for both versions of the problem.

The major contributions of this paper are summarized as follows:
\begin{itemize}
\item We investigate a novel and relevant problem in social networks, the data-driven influence minimization by edge removal. 
\item We study our general problem under both budget and matroid constraints, discussing how these affect algorithmic design. 

\item We show that influence limitation is APX-hard and propose deterministic and probabilistic constant-factor approximation solutions for the budgeted and matroid versions of the problem, respectively.

\item We evaluate the proposed techniques using several synthetic and real datasets. The results show that our methods outperform the baseline solutions by up to $40\%$ while scaling to large graphs.

\end{itemize}

\section{Influence Limitation}
\label{sec:problem_}

We start with a description of Credit Distribution Model and formulate the influence limitation problems in Section \ref{sec:prob_def}.

\subsection{Credit Distribution Model}

The Credit Distribution Model (CDM) \cite{goyal2011} estimates user influence directly from propagation traces. Its main advantages compared to classical influence models (e.g. Independent Cascade and Linear Threshold \cite{kempe2003maximizing,leskovec2007cost}) is that it does not depend on computationally intensive simulations while also relying less on the strong assumptions made by such models. Our algorithms apply CDM to compute user influence, and thus we briefly describe the model in this section.

Let $G(V,E)$ be a directed social graph and $\mathscr{L}(User, Action, Time)$ be an action log, where a tuple $(u,a,t)$ indicates that user $u$ has performed action $a$ at time $t$. Action $a \in \mathscr{A}$ propagates from node $u$ to node $v$ iff $u$ and $v$ are linked in social graph and $u$ performed action $a$ before $v$. This process defines a \textit{propagation graph} (or an \textit{action graph}) of $a$ as a directed graph $G(a)=(V(a),E(a))$ which is a DAG. The action log $\mathscr{L}$ is thus a set of DAGs representing different actions' propagation traces.
For a particular action $a$, a potential influencer of a node or user can be any of its in-neighbours. We denote $N_{in}(u,a)=\{v|(v,u) \in E(a)\}$ as the set of potential influencers of $u$ for action $a$ and $d_{in}(u,a)=|N_{in}(u,a)|$. When a user $u$ performs action $a$, the \textit{direct influence credit}, denoted by $\gamma_{(v,u)}(a)$, is given to all $v\in N_{in}(u,a)$. Intuitively the CDM distributes the influence credit backwards in the propagation graph $G(a)$ such that not only $u$ gives credit to neighbours, but also in turn the neighbours pass on the credit to their predecessors. The total credit, $\Gamma_{v,u}(a)$ given to a user $v$ for influencing $u$ via action $a$ corresponds to multiple paths from $v$ to $u$ in the propagation graph $G(a)$:
\begin{equation}
\label{eq:inf_v_u}
    \Gamma_{v,u}(a)= \sum_{w \in N_{in}(u,a)} \Gamma_{v,w}(a).\gamma_{(w,u)}(a)
\end{equation}
Similarly, one can define the credit for a set of nodes $X$,
\begin{equation*}
\Gamma_{X,u}(a)=\begin{cases}
               1 \quad \quad \quad \quad \quad \quad \quad \quad \quad \quad \quad \quad \quad if\; u \in X\\
               \sum_{w \in N_{in}(u,a)} \Gamma_{X,w}(a).\gamma_{(w,u)}(a) \; \quad otherwise
               \end{cases}
\end{equation*}
By normalizing the total credit over all actions $\mathscr{A}_u$ by a node $u$:
\begin{equation} \label{eq:set_k}
    \kappa_{X,u}=\frac{1}{|\mathscr{A}_u|}\sum_{a\in \mathscr{A}_u}\Gamma_{X,u}(a)
\end{equation}

The total influence $\sigma_{cd}(X)$ is the credit given to $X$ by all vertices:

\begin{equation} \label{eq:cd}
    \sigma_{cd}(G,X)=\sum_{u \in V} \kappa_{X,u}
\end{equation}

\begin{example}
 Figures~\ref{fig:social_ex1} and \ref{fig:dag_ex1} show a social graph $G$ and one propagation graph $G(a)$ for action $a$, respectively. For $G(a)$, $\sigma_{cd}(G,X)=\sum_{u\in V}\Gamma_{X,u}$ given any target set $X$. Consider the following example where $X=\{w,v\}$. In the propagation graph (Fig. \ref{fig:dag_ex1}), $\sigma_{cd}(G,X)= \Gamma_{X,s}+\Gamma_{X,t}+\Gamma_{X,v}+ \Gamma_{X,w}+\Gamma_{X,x}+\Gamma_{X,u}+\Gamma_{X,y}= 0+0+1+1+.5+(.2*.5+.2+.2+.3*1)+1=4.21$. 
 \end{example} 

 Notice that we have assigned the values of $\gamma_{(u',v')}$ arbitrarily. In practice, we compute influence probabilities ($\gamma$) using well-known techniques in \cite{goyal2010learning}. Our theoretical results do not depend on the particular scheme used to compute $\gamma$.

\begin{table} [t]
\centering
\small
\begin{tabular}{| c | c |}
\hline
\textbf{Symbols} & \textbf{Definitions and Descriptions}\\
\hline
$G(V,E)$ & Given graph (vertex set $V$ and edge set $E$)\\
\hline
$X$ & Target set of source nodes\\
\hline
$C$ &The set of candidate edges\\
\hline
$k$ & Budget for BIL\\
\hline
$G(a)=(V(a),E(a))$ & Action/propagation graph (DAG) for action $a$\\
\hline
$\Gamma_{v,u}(a)$ & Credit of node $v$ for influencing $u$ in $G(a)$\\
\hline
$\Gamma_{X,u}(a)$ & Credit given to set $X$ for influencing $u$ in $G(a)$\\
\hline
$\gamma_e(a)=\gamma_{(v,u)}(a)$ & Direct credit for $v$ to influence $u$ via $e=(v,u)$\\
\hline
 $u\overrightarrow{a}v$ & It implies there is a path from $u$ to $v$ in $G(a)$ \\
 \hline
$b$ & Maximum $\#$edges removed from a node in ILM\\
\hline
 $\vec{y}$ &  The vector with edge membership probabilities\\
 \hline
\end{tabular}
\caption {Frequently used symbols}\label{tab:table_symbol}
\end{table}


 \begin{figure}[t]
\small
    \centering
    \subfloat[Social Graph, $G$]{\includegraphics[width=0.15\textwidth]{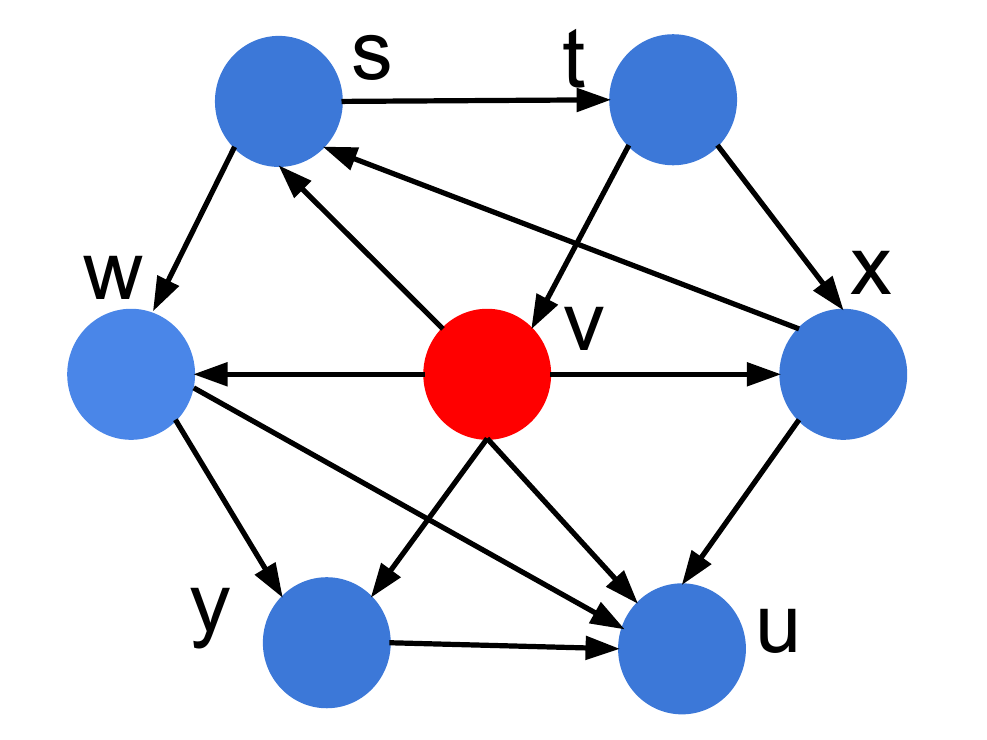}\label{fig:social_ex1}}
    \subfloat[Action Graph, $G (a)$]{\includegraphics[width=0.15\textwidth]{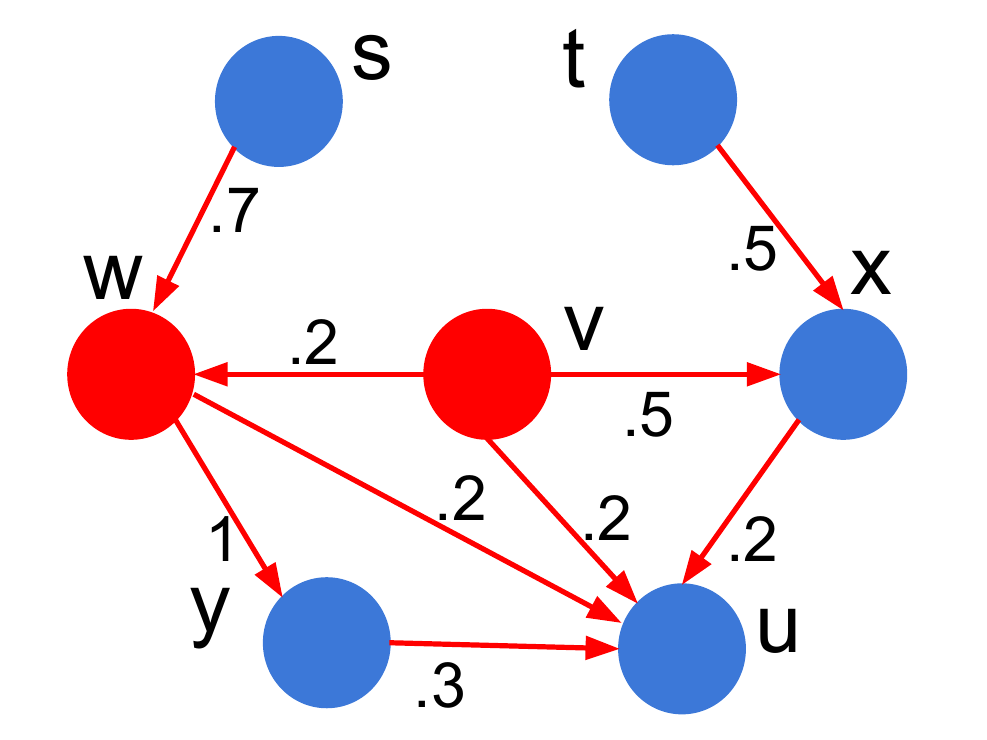}\label{fig:dag_ex1}}
    \subfloat[Modified DAG $G^m (a)$]{\includegraphics[width=0.15\textwidth]{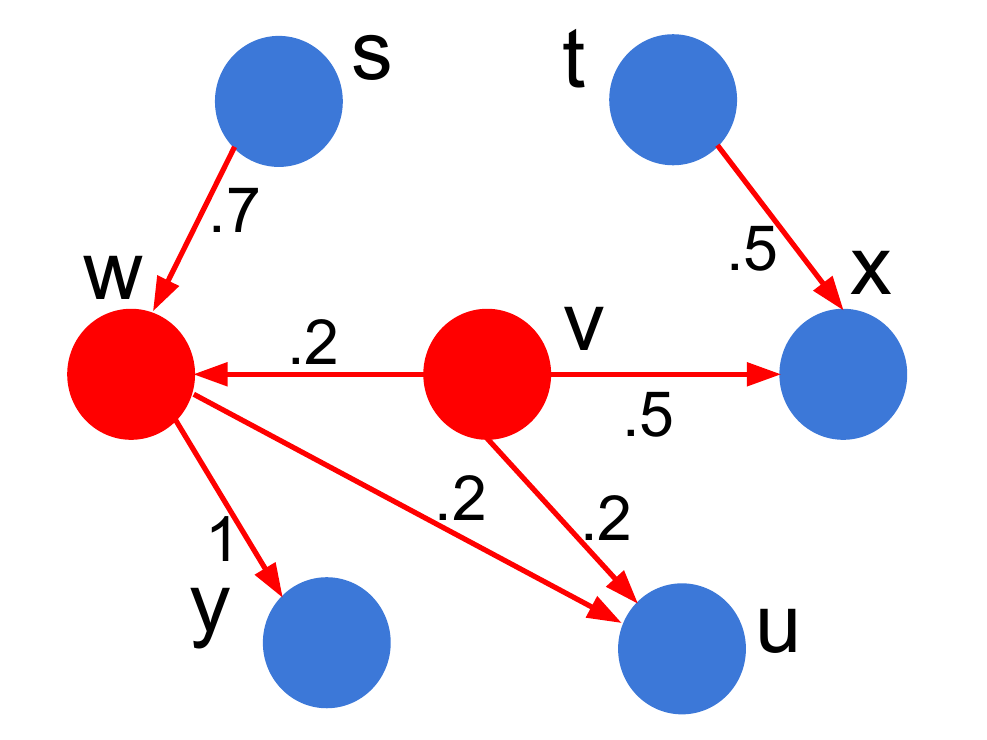}\label{fig:edited_dag_ex1}}
    \caption{Illustrative example of a social graph and CDM with the corresponding credits over the edges. \label{fig:example1}}
\end{figure}
 
 \subsection{Problem Definitions}
 \label{sec:prob_def}

We study influence minimization in two different settings. The first is budget constrained optimization, where a limit on the number of edges to be modified is set as a parameter. The second setting takes into account a more general class of constraints that can be expressed using the notion of an independent set. 


Our goal is to remove a few edges $B\subset E$ such that the influence of a target set of users $X$ is minimized according to the CDM. Given a target user $v$ and an arbitrary user $u$, the credit of $v$ for influencing $u$ in $G(a)$ is computed based on Equation \ref{eq:inf_v_u}. Consider $P(v,u)$ to be the set of paths from $v$ to $u$ where each path $p=\{e_1,e_2,...,e_t\}$ is such that $e_1=(v,v')$, $e_t=(u',u)$, and $e_i\in E(a)$ for all $i$ and $u',v' \in V(a) - \{v,u\}$. We use $\gamma_{(w',w)}(a)$ or $\gamma_{e}(a)$ to represent the credit exclusively via edge $e=(w',w)$ for influencing $w$ in $G(a)$. Therefore, Equation \ref{eq:inf_v_u} can be written as: 
\begin{equation}
    \Gamma_{v,u}(a)=\sum_{p\in P(v,u)} \prod_{e\in p}\gamma_{e}(a)
    \label{eqn::gamma_v_u}
\end{equation}

A similar expression can be defined for a target set $X$:
\begin{equation}
    \Gamma_{X,u}(a)=\sum_{p\in P(X,u)} \prod_{e\in p}\gamma_{e}(a)
    \label{eqn::gamma_x_u}
\end{equation}
where $P(X,u)$ contains only the minimal paths from $v \in X$ to $u$---i.e. $\nexists p_i, p_j \in P(X,u)$ such that $p_i \subseteq p_j$.


We apply Equation \ref{eqn::gamma_x_u} to quantify the change in credit for a target set of nodes $X$ and a particular action $a$ after the removal of edge $e$ according to the credit distribution model:

\begin{equation}
    \delta_a(\{e\})=\sum_{w\in V}\Big(\Gamma_{X,w}(a)-\sum_{\substack{p \in P(X,w)\\e \notin p}}\prod_{e'\in p}\gamma_{e'}(a)\Big)
    \label{eqn::delta_a_e}
\end{equation}

An edge deletion potentially blocks a few paths from $v$ to the remaining users, reducing its credit (or influence). We use $G^m=(V,E-B)$ and $G^m(a)$ to denote the graph and the propagation graph for action $a$ after the removal of edges in $B$, respectively. The following sections, introduce the budget and matroid constrained versions of the influence limitation problem.

\subsubsection{Budgeted Influence Limitation (BIL)} 

We formalize the \textit{Budgeted Influence Limitation (BIL)} problem as follows.

\begin{problem} \label{prob:BIL}
\textbf{Budgeted Influence Limitation (BIL):} 
Given a directed graph $G(V,E)$, an action log $\mathscr{L}$, a candidate set of edges $C$, a given seed set $X$, and an integer $k < |C|$, find a set  $B \subset C \subset E$ of $k$ edges such that $\sigma_{cd}(G^m,X)$ is minimized or, $\Delta(B)= \sigma_{cd}(G,X)-\sigma_{cd}(G^m,X) $ is maximized where $G^m=(V,E\setminus B)$.
\end{problem}

\begin{example}
Consider the example in Figure \ref{fig:example1}, assuming the candidate set $C= \{(t,x), (y,u), (x,u) \}$, $k=2$, and $X=\{w,v\}$. From Example 1, $\sigma_{cd}(G,X)= 4.21$ in the unmodified graph and the deletion of $(t,x) \in C$ will not change the influence of $X$. On the other hand, the removal of $(y,u)$ and $(x,u)$ (Fig. \ref{fig:edited_dag_ex1}) will make $\sigma_{cd}(G^m(a),X)=\Gamma_{X,v}+\Gamma_{X,w}+\Gamma_{X,x}+\Gamma_{X,u}+\Gamma_{X,y}=1+1+ 0.5+(0.2+0.2)+1=3.9$. 

\end{example} 

We show that our problem is NP-hard.

\begin{thm}

\label{thm:hardness} The BIL problem is NP-hard.
\end{thm}
\begin{proof}
See the Appendix.
\end{proof}

BIL assumes that any $k$ edges in the candidate set can be removed from the network.  While such budget constrained formulations are quite popular in the literature \cite{goyal2011,kempe2003maximizing,Khalil2014}, they fail to capture relevant aspects in many applications. For instance, an optimal solution for BIL might make the network disconnected or disproportionately affect particular sub-networks. Table \ref{table:mot_exp_matroid} exemplifies this issue using two real networks and different sizes of the target set $X$ chosen uniformly at random. The majority of the modifications are concentrated in the top three nodes---i.e. those with the largest number of edges removed in the solution.
In the next section, we present a different formulation for influence limitation under matroid constraints, which addresses some of these challenges.

\subsubsection{Influence Limitation under Matroid (ILM)} \label{subsec::ilm}

\textit{Matroids} are abstract objects that generalize the notion of linear independence to sets \cite{chekuri2010dependent}. We apply matroids to characterize a class of constraints for influence limitation. 
 First, we formalize the concept of a matroid:

\begin{definition} \label{def:matroid}
\textbf{Matroid \cite{nemhauser1978}:} A finite matroid  $M$ is a pair $(C,I)$, where $C$ is a finite set (called the ground set) and $I$ is a family of subsets (independent sets) of $C$ with the following properties:
\begin{enumerate}
    \item The empty set is independent, i.e., $\emptyset \in I$.
    \item Every subset of an independent set is independent.
    \item If $M$ and $N$ are two independent sets of $I$ and $|M|>|N|$, then there exists $x \in M\setminus N$ such that $N\cup \{x\} \in I$.
\end{enumerate}
\end{definition}

To illustrate the expressive power of matroids as a general class of constraints for optimization problems defined over networks, we focus on a particular setting of influence minimization. More specifically, we upper bound the number of edges that can be removed from each node in the network. 

\begin{problem}
 \label{prob:ILM}
\textbf{Influence Limitation under Matroid (ILM):} Given a directed social graph $G(V,E)$, an action log $\mathscr{L}$, a candidate set of edges $C$, a given seed set $X$, and an integer $b$, find a set $B$ $($where $B \subset C \subset E)$ such that at most $b$ edges from $B$ are incident (incoming) on any node in $V$ and $\sigma_{cd}(G^m,X)$ is minimized where $G^m=(V,E - B)$ or, $\Delta(B)= \sigma_{cd}(G,X)-\sigma_{cd}(G^m,X) $ is maximized.
\end{problem}

The effect of ILM is to enforce network modifications that are more uniformly distributed across the network. Notice that a valid solution for the budget constrained version (BIL) might not necessarily be a valid solution for ILM. Conversely, not every solution of ILM is valid for BIL. We also show that ILM is NP-hard.

\begin{thm}\label{thm:hardness_ilm} The ILM problem is NP-hard.
\end{thm}
\begin {proof}
The proof follows a similar construction as in Thm. \ref{thm:hardness}.
\end{proof}

 \begin{table}[t]
\centering
\begin{tabular}{| c | c | c | c | c|}
\hline
&\multicolumn{2}{c|}{\textbf{CA}} & \multicolumn{2}{c|}{\textbf{FXS}} \\
\hline
&  $|X|=20$ &  $|X|=30$ &  $|X|=20$ &  $|X|=30$ \\
\hline
Round $1$ & $78$ & $68$ & $60$ & $55$  \\
\hline
Round $2$ & $76$  & $66$ &  $55$ & $50$ \\
\hline
Round $3$ &   $78$ & $64$ & $58$ &   $49$ \\
\hline
Round $4$ &  $70$ & $70$ & $64$ &  $52$\\
\hline
Round $5$ &  $68$ & $ 68$ & $63$ &  $55$\\
\hline
\end{tabular}

\caption{Motivation for influence limitation under matroid (ILM). We compute the percentage of removed edges, from a total of 50, that are incident to the top three nodes in the solution of the budgeted version of the problem (BIL). In each round, we select a target set $X$ uniformly at random.  Results from two datasets (CA and FXS, see description in Section \ref{sec:exp}) are shown. Notice that BIL modifications are strongly biased towards a small set of nodes in the network. Using a matroid constraint (ILM), we are able to enforce modifications that are better distributed across the network. \label{table:mot_exp_matroid}}

 \end{table}

It remains to show that ILM follows a matroid constraint---i.e. any valid solution for ILM is a matroid (Definition \ref{def:matroid}). In fact, we will show that ILM follows a \textit{partition matroid}, which is a specific type of a matroid where the ground set $C$ is partitioned into non-overlapping subsets $C_1,C_2,\cdots, C_l$ with associated integers $b_1,b_2,\cdots,b_l$ such that a set $B$ is independent iff $|B\cap C_i| \leq b_i$.


\begin{observation}
\label{obs:matroid}
ILM follows a partition matroid.
\end{observation}

 The key insight for this observation is that, for any incoming edge, the associated node is unique to the edge. As an example, if $e=(u,v)$ (incoming to $v$) then the node $v$ is unique to the edge $e$. Thus, the ground set $C$ can be partitioned into edge sets $(C_1,C_2,..., C_{|V|})$ based on the $|V|$ unique incidence edges associated with them. Any feasible solution $B$ (edge set) is an independent set as $B\cap C_v\leq b$, where $v \in V$. Notice that the more general setting where a constant $b_v$ is defined for each node in the network is also a partition matroid.

\section{Submodularity}\label{sec:edge_contribution}

A key feature in the design of efficient algorithms for influence limitation is \textit{submodularity}. Intuitively, submodular functions are defined over sets and have the so called \textit{diminishing returns property}. These functions behave similarly to both convex and concave functions \cite{krause2014submodular}, enabling a polynomial-time search for approximate global optima for NP-hard problems. Besides its more usual application to the budgeted version of our problem, we also demonstrate the power of submodular optimization in the solution of influence limitation problems under matroid constraints. 

In order to prove that the maximization function $\Delta$ associated to both BIL and ILM is submodular, we analyze the effect of the removal of a single candidate edge $e$ over the credit of the target set $X$. Equation \ref{eqn::delta_a_e} defines the change in credit ($\delta_{a} (\{e\})$) after the removal of $e=(u,v)$ in $G(a)$. In case a given vertex $v$ does not have outgoing edges in $G(a)$, the change can be computed as:

\begin{equation*}
\delta_{a} (\{e\}) = \big ( \Gamma_{X,u}(a). \gamma_{(u,v)}(a) \big )  \Gamma_{v,v}(a)  
\end{equation*}

The next lemma describes the effect of removing an edge $e=(u,v)$ for the case where node $v$ has outgoing edges in $G(a)$. 

\begin{lemma}\label{lemma:credit_computation2}
For an action $a$, with corresponding DAG $V(a)$, the change in credit after the removal of $e=(u,v)$ is as follows:

\begin{equation}
     \delta_{a} (\{e\}) = \big(\Gamma_{X,u}(a)\cdot \gamma_{(u,v)}(a)\big)\cdot \sum_{w \in V} \Gamma_{v,w}(a)
\end{equation}

\end{lemma}


\begin{proof}
The proof is based on induction over the length of the paths from $v$ to $w$ (see the Appendix).
\end{proof} 

\begin{example}
Consider the example in Figure \ref{fig:dag_ex1}. Let the target set $X$ be $\{v\}$.
The contribution of the edge $(w,y)$ will be the following: 
$\big(\Gamma_{v,w}\cdot \gamma_{(w,y)}\big)\cdot (\Gamma_{y,u}+\Gamma_{y,x}+\Gamma_{y,t}+\Gamma_{y,s})$. Now, $\Gamma_{y,u}=0.3,\Gamma_{y,x}=0,\Gamma_{y,t}=0, \Gamma_{y,s}=0$ and $\Gamma_{v,w}=0.2$. So, the marginal contribution of the edge $(w,y)$ is $(0.2)\cdot 1 \cdot (0.3+0+0+0)= 0.06$.
\end{example}

We are now able to formalize the change in credit due to a single edge deletion over all the actions in the action set $\mathscr{A}$.

\begin{lemma}
\label{lemma:edge_contribution}
The total change in credit $\Delta(\{e\})$ due to the removal of edge e can be computed as:
\begin{equation*}
\begin{split}
\Delta(\{e\}) & = \sigma_{cd}(G,X)-\sigma_{cd}(G^m,X)\\
& = \sum_{a \in \mathscr{A}} \big(\Gamma_{X,u}(a). \gamma_{(u,v)}(a)\big) . \big(\sum_{w \in V} \frac{1}{ |\mathscr{A}_w|}\Gamma_{v,w}(a)\big)
\end{split}
\end{equation*}
where $G^m=(V,E\setminus \{e\})$.
\end{lemma}

Lemma \ref{lemma:edge_contribution} follows from Lemma \ref{lemma:credit_computation2} and Equations  \ref{eq:set_k}, and \ref{eq:cd}. Next, we prove the submodularity property of the function $\Delta$.

\begin{thm}\label{thm:submodular} The function $\Delta$  is  monotone and submodular.
\end{thm}
\begin{proof}
The function is monotonic for each action $a$, as the removal of an edge cannot increase the credit. 
As a consequence, $\Delta$ which is a sum of credits over all actions is also monotonic. 

To prove submodularity, we consider the deletion of two sets of edges, $E_S$ and $E_T$ where $E_S \subset E_T$, and show that $\Delta(E_S \cup \{e\}) - \Delta(E_S) \geq \Delta(E_T\cup \{e\})-\Delta(E_T)$ for any edge $e \in C$ such that $e \notin E_S$ and $e \notin E_T$. A non-negative linear combination of submodular functions is also submodular. Thus, it is sufficient to show the property for one action $a$, as $\Delta$ has the following form:  
\begin{equation*}
\begin{split}
\Delta(B)&= \sigma_{cd}(G,X)-\sigma_{cd}(G^m,X) \\
&=\frac{1}{|\mathscr{A}_u|}\sum_{a\in \mathscr{A}_u}\Gamma'_{X,u}(G,a)- \frac{1}{|\mathscr{A}_u|}\sum_{a\in \mathscr{A}_u}\Gamma'_{X,u}(G^m,a) \\ &=\frac{1}{|\mathscr{A}_u|}\sum_{a\in \mathscr{A}_u}(\Gamma'_{X,u}(G,a)-\Gamma'_{X,u}(G^m,a))
\end{split}
\end{equation*}
where $\Gamma'_{X,u}(G,a)$ denotes $\Gamma_{X,u}(a)$ in $G(a)$.

For the same reason, we assume a single node $x \in X$ ($\Gamma_{X,u}= \sum_{s\in X} \Gamma_{s,u}^{V-X+s}$). Edge sets $E_S$ and $E_T$ are removed from the graph and we evaluate $\Delta(\{e\})$ such that $e \notin E_S$ and $e \notin E_T$. Let the credits towards $x$ from node $w$ after removing $E_S$ and $E_T$ edges be $\Gamma'_{x,w}(G^S)$ and $\Gamma'_{x,w}(G^T)$ (omitting $a$ from $\Gamma'(.,a)$ for simplicity) respectively. Moreover, use the notation $u\overrightarrow{a}v$ if there is a path from $u$ to $v$ in $G(a)$. There are two possible cases. 

1) If  $w\overrightarrow{a}v$ does not hold, then removal of $e=(u,v)$ keeps $\Gamma'_{x,w}(G^S)$ and $\Gamma'_{x,w}(G^T)$ unchanged. Hence the marginal gains due to $e$ for both $E_S$ and $E_T$ are $0$. 

2)  If $w\overrightarrow{a}v$ holds, marginal gains for sets $E_S$ and $E_T$ are equal to $\Gamma'_{x,u}(G^S). \gamma_{(u,v)}. \Gamma'_{v,w}(G^S)$ and  $\Gamma'_{x,u}(G^T). \gamma_{(u,v)}. \Gamma'_{v,w}(G^T)$.

Thus, $\Gamma'_{x,u}(G^S). \gamma_{(u,v)}. \Gamma'_{v,w}(G^S) \geq \Gamma'_{x,u}(G^T). \gamma_{(u,v)}. \Gamma'_{v,w}(G^T)$ as $\Gamma'_{x,u}(G^S)\geq \Gamma'_{x,u}(G^T)$ and $\Gamma'_{v,w}(G^S)\geq \Gamma'_{v,w}(G^T)$. This shows that $\Delta$ is a submodular function.
\end{proof}

The next two sections describe how we apply the submodularity property in the design of efficient approximate algorithms for influence minimization (BIL and ILM).



\section{Budget Constrained Problem} \label{sec:budget_BIL}
According to Theorem \ref{thm:submodular}, BIL is a monotone submodular maximization problem under a budget constraint. As a consequence, a simple greedy algorithm produces a  constant factor approximation of $(1-1/e)$~\cite{nemhauser1978} for the problem. However, naively applying greedy algorithm might be expensive. It requires scanning the action log file, computing the credits and updating them multiple times after each edge removal. We introduce a more efficient version of this greedy algorithm based on properties of the credit distribution model.

\begin{algorithm}[t]
\caption{Greedy}
\begin{algorithmic}[1] 
 \REQUIRE $X$, $C$, $k$
\ENSURE A solution set $B$ of $k$ edges 
\STATE $B\leftarrow\emptyset$
\WHILE {  $|B|\leq k$ }
\FOR{$e \in C \setminus B$}
\STATE \textit{e.MC}$\leftarrow$ computeMC($e$)
\ENDFOR
\STATE $e^*\leftarrow \argmax_{e\in C \setminus B}\{e.MC\}$
\STATE $B\leftarrow B\cup \{e^*\}$ and $E\leftarrow E\setminus \{e^*\}$ 
\STATE updateUC($e,EP,UC,SC$)
\STATE updateSC($e,EP,UC,SC$)
\ENDWHILE
\end{algorithmic}
\label{algo:Greedy}
\end{algorithm}

\begin{algorithm}[t]
\caption{computeMC}
\begin{algorithmic}[1] 
 \REQUIRE $e=(u,v)$, $X$, $UC$, $SC$
\ENSURE $mc$
\STATE $mc=0$
\FOR{$a \in \mathscr{A}$ such that $SC[u][a]>0$ and $EP[u][v][a]>0$}
\STATE $mc_a=0$
\FOR{each user $w$ such that $UC[v][w][a]>0$}
\STATE $mc_a=mc_a+UC[v][w][a]/\mathscr{A}_w$
\ENDFOR
\STATE $mc=mc+(SC[u][a]\cdot EP[u][v][a])\cdot mc_a$
\ENDFOR
\end{algorithmic}
\label{algo:computeMC}
\end{algorithm}


The greedy algorithm removes the edge that minimizes the credit (or influence) of the target set, one at a time. After each edge removal, the credit $\Gamma_{u,v}$, i.e., the credit of node $u$ for influencing $v$, has to be updated. As the algorithm removes only one edge $e$, intuitively, it should not affect nodes in the entire network but only some in the neighborhood of $e$. Next, we formalize these observations and show how to apply them in the design of an efficient algorithm for BIL. 

\begin{observation} \label{observation:edge_removal1}
For a given action $a$ and DAG $G(a)$, the removal of $e=(u,v)$ changes $\Gamma_{z,w}$ iff $z\overrightarrow{a} u$ and $v\overrightarrow{a} w$. 
\end{observation}
Let us consider an arbitrary DAG $G(a)$ and node pair $(z,w)$. Deleting $e=(u,v)$ can only affect the credit $\Gamma_{z,w}$---i.e., the credit of node $z$ for influencing $w$---if $e$ is on an path from $z$ to $w$ in $G(a)$. The edge $e$ is on one of such paths if and only if $z\overrightarrow{a}  u$ and $v\overrightarrow{a}  w$. 

The following observations can be derived from Lemma \ref{lemma:edge_contribution}.

\begin{observation} \label{obs:edge_removal2}
For given action $a$ and DAG $G(a)$, the removal of $e=(u,v)$ reduces $\Gamma_{z,w}$ by $(\Gamma_{z,u}\cdot\gamma_{(u,v)})\cdot\Gamma_{v,w}$ iff $z\overrightarrow{a}  u$ and $v\overrightarrow{a} w$.
\end{observation}

\begin{observation} \label{obs:edge_removal3}
For given target set $X$, an action $a$ and DAG $G(a)$, the removal of $e=(u,v)$ reduces $\Gamma_{X,w}$ by $(\Gamma_{X,u}\cdot\gamma_{(u,v)})\cdot\Gamma_{v,w}$ iff $z\overrightarrow{a}  u$ and $v\overrightarrow{a} w$ where $z\in X$.
\end{observation}

Next we describe our algorithms for the BIL problem.


Algorithm \ref{algo:Greedy} scans the actions log  $\mathscr{L}$ to collect information for comparing the effect of removing each candidate edge. This information is maintained in data structures EP, EC, and SC. In particular, $EP[u][v][a]$ denotes the edge credit ($\gamma_{(u,v)} (a)$) of $u$ for influencing $v$ when $(u,v)$ exists, $UC[u][v][a]$ is the credit ($\Gamma_{u,v} (a)$) given to $u$ for influencing $v$, and $SC[u][a]$ is the credit ($\Gamma_{X,u} (a)$) given to $X$ for influencing $u$, all for  an action $a$.

The contribution of each edge (see Lemma \ref{lemma:edge_contribution}), given the current solution set $B$, is computed using Algorithm \ref{algo:computeMC}. Methods \textit{updateUC} and \textit{updateSC} are based on observations \ref{obs:edge_removal2} and \ref{obs:edge_removal3}, respectively. While \textit{updateUC} updates UC upon an edge removal, \textit{updateSC} updates the credit of the target set in SC (see the Appendix for details).


The three most expensive steps of the greedy algorithm are steps $4$, $8$ and $9$. The corresponding methods \textit{computeMC}, \textit{updateUC}, and \textit{updateSC}, take $O(\sum_{a \in \mathscr{A}} |V(a)| )$, $O(\sum_{a \in \mathscr{A}} |V(a)|^2 )$, and  $O(\sum_{a \in \mathscr{A}} |V(a)| )$ time, respectively. Thus, the total running time of Greedy is $O(k\cdot |C| \cdot \sum_{a \in \mathscr{A}} |V(a)| + k\cdot \sum_{a \in \mathscr{A}} |V(a)|^2)$. Notice that the time does not depend on the number of nodes in the graph ($|V|$), but on the size of the action graphs, the budget and the size of the candidate edge set. We discuss further optimization techniques for Algorithm \ref{algo:Greedy} in the Appendix.


\section{Matroid Constrained Problem }
\label{sec:ilm_algorithms}

In the previous section, we have presented an efficient greedy algorithm for budgeted influence limitation. Here, we switch to the matroid constrained version of the problem (ILM), for which the described algorithm (Algorithm \ref{algo:Greedy}) might not provide a valid solution. Notwithstanding, as for the budgeted case, the submodularity of the influence minimization objective (see Section \ref{sec:edge_contribution}) plays an important role in enabling the efficient solution of ILM. 

Based on Observation \ref{obs:matroid} (Section \ref{subsec::ilm}), we apply existing theoretical results on submodular optimization subject to matroid constraints in the design of our algorithm. First, we propose a continuous relaxation that is the foundation of a continuous greedy algorithm for ILM. Next, we describe two techniques for rounding the relaxed solution. While the first rounding scheme also achieves an approximation factor of $1-1/e$, it is not scalable. Thus, we propose a faster randomized rounding scheme that we will show to work well in practice. Generalizations and hardness results based on the notion of curvature \cite{vondrak2010submodularity} are covered in Sections \ref{sec:generalizations} and \ref{sec:curvature}.

\subsection{Continuous Relaxation} \label{sec:continuous_relaxation}

Let $\vec{y}=(y_1,y_2,...y_c)$ be the vector with membership probabilities for each edge in the candidate set, $C$ ($|C|=c$). Moreover, let $B$ be a random subset of $C$ where the edge $e_i\in C$ is included in $B$ with probability $y_i$. From \cite{Vondrak08}, if $f$ is the continuous extension of $\Delta$, then:
\begin{equation}
f(\vec{y})= \textbf{E}_{B\sim \vec{y}}[\Delta(B)] = \sum_{B\subseteq C} \Delta(B) \prod_{e_i\in B}{y_i}\prod_{e_i\in C\setminus B}{(1-y_i)}
\end{equation}

Let $E_{in}(v)$ be the set of incoming edges to the node $v$. Our objective is to find a $\vec{y}$ that maximizes $f(\vec{y})$ with the following constraints:
\begin{equation}
\label{eq:prob_edge}
    y_i\in [0,1]
\end{equation}
\begin{equation}
\label{eq:constraint_node}
\sum_{e_i\in E_{in}(v)} y_i \leq b\quad \forall v \in V
\end{equation}

While Equation \ref{eq:prob_edge} maintains the fractional values as probabilities, Equation \ref{eq:constraint_node} enforces the maximum number of edges incident to each node to be bounded by $b$. Because the relaxation of $\Delta$ as $f$ is continuous, the optimal value for $f$ is an upper bound on $\Delta$ (the discrete version). Let $B^*$ and $Y^*$ be the optimal edge sets for $\Delta$ and $f$, respectively. Also, let $Z$ be a vector defined as follows: $z_i=1$ if $e_i\in B^*$ and $z_i=0$, otherwise. Then, $\Delta(B^*)=f(Z)$ and $Z$ maintains the constraints. As $f(Y^*)$ is maximum, $\Delta(B^*)=f(Z) \leq f(Y^*)$.

\begin{algorithm}[t]
\caption{Continuous Greedy (CG)}
\begin{algorithmic}[1] 
 \REQUIRE $X$, $C$, $b$
\ENSURE A vector $\vec{y}$ satisfying Eqs. \ref{eq:prob_edge} and \ref{eq:constraint_node}
\STATE Start $\vec{y}$ as a null vector, $t=0$
\WHILE {  $t\leq \tau $ }
\STATE Generate $s$ samples $B_1,B_2,...,B_s$ where $e_i$ belongs to $B_j$ ($\forall j\in [s]$) with probability $y_i$
\STATE Set weight of an edge, $e_i$ as $w_i=\frac{\sum_{j=1}^s{\Delta(B_j\cup\{e_i\})-\Delta(B_j)}}{s}$
\STATE Compute an edge set $E^Y$ maintaining the constraint (Eq. \ref{eq:constraint_node}) and maximizes  $\sum_{e_i\in E^Y} w_i$
\STATE For all $e_i\in E^Y$, set $y_i=y_i+1/\tau$
\STATE $t=t+1$
\ENDWHILE
\RETURN $\vec{y}$
\end{algorithmic}
\label{algo:CG}
\end{algorithm}


We show that the new objective function $f$ is smooth (i.e. it has a second derivative), monotone and submodular. Based on these properties, we can design a continuous greedy algorithm that produces a relaxed solution for ILM with a constant-factor approximation  \cite{Vondrak08}.

\begin{thm}\label{thm:continuous_submodular}
The objective function $f$ is a smooth monotone submodular function.
\end{thm}

\begin{proof}
The proof exploits monotonicity and submodularity of the associated function $\Delta$ (see the Appendix).
\end{proof}


\textbf{Continuous Greedy (CG): } The continuous greedy algorithm (Algorithm \ref{algo:CG})  provides a solution set $\vec{y}$ such that $f(\vec{y})\geq(1-\frac{1}{e})f(Y^*)\geq (1-\frac{1}{e})\Delta(B^*)$ with high probability. The approximation guarantee exploits the facts that $\Delta$ is submodular (Theorem \ref{thm:continuous_submodular}) and ILM follows a matroid constraint (Observation \ref{obs:matroid}). CG is similar to the well-known Frank-Wolfe algorithm \cite{nocedal2006numerical}. It iteratively increases the coordinates (edge probabilities) towards the direction of the best possible solution with small step-sizes while staying within the feasible region. 
In \cite{Vondrak08}, Vondrak proves the following:
\begin{thm}\label{thm:vondrak_thm} 
The Continuous Greedy (Algorithm \ref{algo:CG}) returns a vector $\vec{y}$ that satisfies Equations \ref{eq:prob_edge} and \ref{eq:constraint_node} and such that $f(\vec{y})\geq (1-\frac{1}{e})\Delta(B^*)$ when $\tau=c^2$ and $s=c^5$.
\end{thm}

The values $\tau$ and $s$ correspond to the number of iterations and samples applied by the CG algorithm. 

The costliest operations of CG are steps 3, 4 and 5. Step 3 takes $O(c\cdot s)$ time, as it visits each edge in the candidate set $C$ ($|C|=c$). Step 4 computes the contribution of edges, having worst case time complexity $O(s\cdot c \cdot \sum_{a \in \mathscr{A}} |V(a)|)$. Step 5 greedily selects the best set of edges, according to the weights. Therefore, the total running time of the algorithm is $O(\tau (s\cdot c \cdot \sum_{a \in \mathscr{A}} |V(a)|+  c\log c))$. 



\subsection{Rounding} \label{sec:rounding}
Algorithm \ref{algo:CG} returns a vector $\vec{y}$ satisfying Equations \ref{eq:prob_edge} and \ref{eq:constraint_node} while producing $f(\vec{y})\geq (1-\frac{1}{e})\Delta(B^*)$. However, as the vector $\vec{y}$ contains values as probabilities (between $0$ and $1$), a rounding step over the vector $\vec{y}$ is still required for obtaining a deterministic set of edges. 

We describe two rounding schemes. The first is a computationally-intensive lossless rounding procedure for matroids known as \textit{swap rounding} \cite{chekuri2010dependent}. Next, we address the high time complexity issue, by proposing a simpler and faster randomized procedure. We show that, our independent rounding method produces feasible edges with low error and high probability. 

\subsubsection{Dependent Rounding \cite{chekuri2010dependent}:} The main idea of this technique is to represent the solution as a linear combination of maximal independent sets in the matroid. After obtaining the representation, the \textit{strong exchange property} \cite{chekuri2010dependent} of matroids is applied in a probabilistic way to generate the final solution. More details are given in \cite{chekuri2010dependent}. 

\subsubsection{Randomized Rounding \cite{raghavan1987randomized}: }  
We sort the edges according to their weights (probabilities) and round them while maintaining feasibility. The procedure is fast as it only makes a single pass over the candidate edges in $C$. In practice, we perform this $50$ times and choose the best solution among the rounded feasible sets. In Section \ref{sec:ilm_exp}, we show that this procedure generates good results.

In order to analyze the effect of this randomized procedure, we assume that it is unaware of the dependency between the edges. Let $\mathcal{B}$ be the edge set produced by rounding, i.e. $f(\vec{y})= \textbf{E}[\Delta(\mathcal{B})]$, and let $E_v \subset \mathcal{B}$ be the incoming edges incident on node $v$. The next theorem shows that the randomized procedure will produce a feasible set within error $\epsilon$ with (high) probability $1-\frac{1}{n}$, where $n=|V|$ is the number of nodes.

\begin{thm}\label{thm:randomized_rounding}
The following bound holds for the number of edges incoming to $v$ in the rounded set:
\begin{equation*}
Pr(|E_v|< (1+\epsilon)b)\geq 1- \frac{1}{n}
\end{equation*}
 where $\epsilon= \sqrt{\frac{6\log{n}}{b}}$.
\end{thm}

\begin{proof}
Let $\mathcal{B}$ be the set of edges produced by the rounding procedure. An edge $e_i$ is included in $\mathcal{B}$ with probability $y_i$. As $\vec{y}$ is a feasible solution, $\sum_{e_i} y_i \leq b \quad \forall v \in V$ (Equation \ref{eq:constraint_node}) where $e_i$ is incident (incoming) to vertex $v$. Thus, $\textbf{E}(|E_v|)\leq b$. Applying the Chernoff's bound:

\begin{equation*}
Pr(|E_v| \geq (1+\epsilon)\textbf{E}(|E_v|)) < \exp{(-\frac{\textbf{E}(|E_v|)\epsilon^2}{3})}
\end{equation*}
Applying the union bound, $\forall v \in V$, we get: 
\begin{equation*}
Pr(|E_v| \geq (1+\epsilon)b) < n\cdot \exp{(-\frac{b\epsilon^2}{3})}
\end{equation*}

Substituting $\epsilon= \sqrt{\frac{6\log{n}}{b}}$, we get:
\begin{equation*}
Pr(|E_v| < (1+\epsilon)b) \geq 1- \frac{n}{n^2} = 1- \frac{1}{n}
\end{equation*}
This ends the proof.
\end{proof}

We emphasize two implications of this theorem: (1) The probability that the rounded solution is feasible depends on the error $\epsilon$ which is small whenever $b$ is large; (2) The rounding procedure has a probabilistic bi-criteria approximation, being lossless if the maximum number of edges to be removed per node is $b'=b(1+\epsilon)$.  

The proposed randomized rounding scheme is efficient, as it only performs one pass over the candidate edges $C$ in order to generate its output. In the Appendix, we compare the performance of the dependent and randomized rounding schemes.

\subsection{Generalizations} \label{sec:generalizations}

We briefly discuss other relevant scenarios where matroid constrained optimization can be applied in the context of influence limitation. 
Matroids can capture a large number of influence limitation settings, especially when edges in the solution can be naturally divided into partitions. Examples include the limitation of influence in non-overlapping communities \cite{bozorgi2017community}, disjoint campaigning \cite{lake1979new}, and problems where issues of fairness arise \cite{yadav2018bridging}. Moreover, influence boosting problems via attribute-level modification \cite{lin2017boosting} and edge addition \cite{Khalil2014} can also be modelled under matroid constraints. 
\subsection{Curvature and APX-hardness} \label{sec:curvature}
\begin{figure}[t]
\small
    \centering
    {\includegraphics[width=0.4\textwidth]{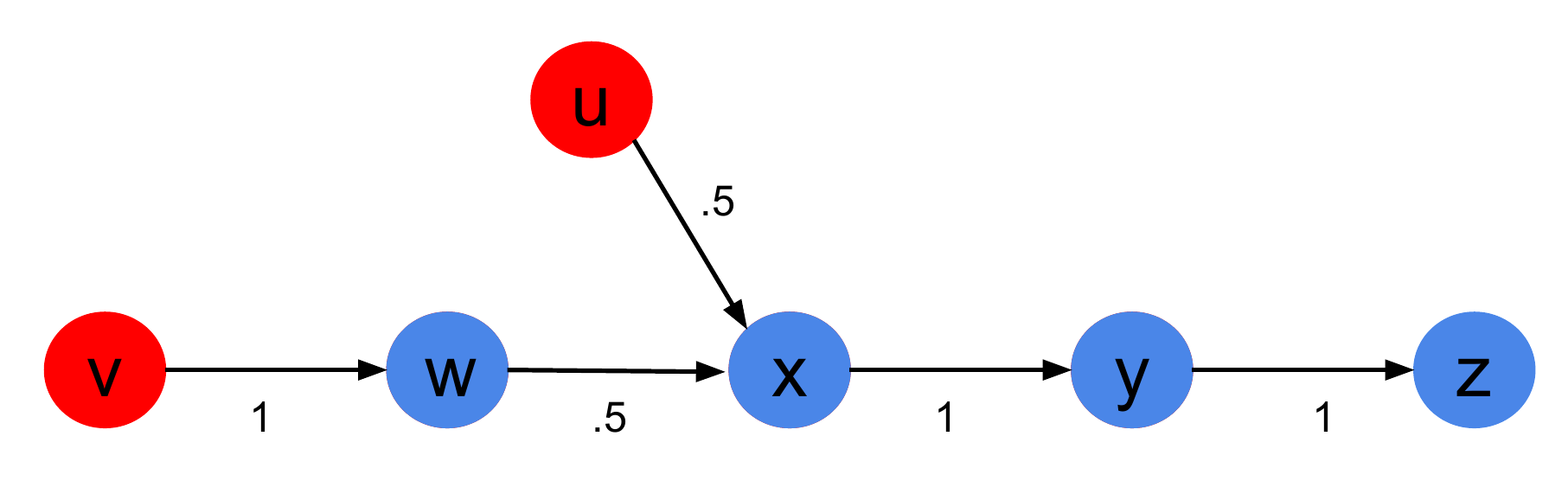}}
    \caption{ This illustrates a counter example in Theorem \ref{thm:apxhardness_ILM}. \label{fig:matroid_ex}}
\end{figure}

The ILM problem is NP-hard to approximate within a constant greater than $1-\frac{1}{e}$. We prove the same about BIL (budget constrained) in the Appendix. To show these results, we first describe a parameter named \textit{curvature} \cite{iyer2013curvature} that models the dependencies between elements (edges) in maximizing an objective function. 

In ILM, the objective is $max\{\Delta(B), B \subset C\}$ where $B$ is an independent set (Definition \ref{def:matroid}). Before proving APX-hardness, we first define the concept of \textit{total curvature} ($c_t$) \cite{vondrak2010submodularity}.

\begin{definition}
The total curvature of a monotone and submodular function  $\Delta$ is defined by:
\begin{equation*}
    c_t= 1-min_{S,e_i}\frac{\Delta(S\cup\{e_i\})-\Delta(S)}{\Delta(\emptyset\cup\{e_i\})-\Delta(\emptyset)}
\end{equation*}
\end{definition}

 The \textit{total curvature} measures how much the marginal gains decrease when an element is added to a set $S$. Intuitively, it captures the level of dependency between elements in a set $S$. For instance, if the marginal gains are independent ($c_t=0$) a simple greedy algorithm will be optimal. Let $S^*$ be the optimal solution set. The  \textit{curvature with respect to optimal} ($c_o$) \cite{vondrak2010submodularity} is defined as follows: 
 \begin{definition}
$\Delta$ has \textit{curvature with respect to optimal} $c_o\in[0,1]$ if  $c_o$ is the smallest value such that for every $T$:
\begin{equation*}
    \Delta(S^*\cup T) - \Delta(S^*)+ \sum_{j\in S^*\cap T} \big(\Delta(S^*\cup T\setminus\{e_i\}) -\Delta(S^*\cup T)\big) \geq (1-c_o)\Delta(T) 
\end{equation*}
\end{definition}

Vondrak \cite{vondrak2010submodularity} proves that there is no polynomial time algorithm that generates a better approximation than $\frac{1}{c_o}(1-e^{-c_0})$ for maximizing a monotone and submodular function with curvature $c_o$ under matroid constraints.

\begin{thm} \label{thm:apxhardness_ILM}
ILM is APX-hard and cannot be approximated within a factor greater than $(1-1/e)$.
\end{thm}

\begin{proof}
ILM is a monotone and submodular optimization problem under a matroid constraint. We prove the inapproximability result by designing a problem instance where the \textit{curvature with respect to optimal} ($c_o$) is $1$. Consider the example in Figure \ref{fig:matroid_ex}, the candidate set $C=\{(w,x),(x,y),(y,z)\}$, $b=1$ and the target set $X=\{u,v\}$. In this setting, one of the optimal sets $S^*={(w,x),(x,y)}$. Assuming $T={(y,z)}$ will imply $S^*\cap T=\emptyset$. If $ \Delta(S^*\cup T) - \Delta(S^*)=0$, then $c_o$ has to be $1$. Note that, $ \Delta(S^*\cup T) = \Delta(S^*)= 2.5$, which leads to $c_o=1$. Therefore, ILM cannot be approximated within a factor greater than $\frac{1}{1}(1-e^{-1})$ and our claim is proved.
 \end{proof}

Both the BIL and ILM problems are APX-hard and cannot be approximated within a  constant greater than $1-\frac{1}{e}$. However, Algorithm \ref{algo:Greedy} (Greedy) provides tight approximation ($1-\frac{1}{e}$) for BIL and Algorithm \ref{algo:CG} (CG) produces the same for ILM with high probability.

\section{Experimental Results}
\label{sec:exp}

\begin{table}[t]
\centering
\begin{tabular}{| c|c |c | c | c |}
\hline
\textbf{Dataset Name}& \textbf{$|V|$} & \textbf{$|E|$} & $\#$Action &$\#$Tuple\\
\hline
\textbf{ca-AstroPh (CA)}& $18k$&$197k$ & $1k$ &$56k$\\
\hline
\textbf{email-EuAll (EE)}& $265k$&$420k$ & $-$ &$-$\\
\hline
\textbf{Youtube (CY)}& $1.1m$&$2.9m$ & $-$ &$-$\\
\hline
\textbf{Flixster-small (FXS)}& $15k$&$191k$ & $1.8k$ &$30k$\\
\hline
\textbf{Flickr-small (FCS)}& $15k$&$1.4m$ & $1.4k$ &$10k$\\
\hline
\textbf{Flixster (FX)}& $1m$&$28m$ & $49k$ &$8.2m$\\
\hline
\textbf{Flickr (FC)}& $1.3m$&$81m$ & $296k$ &$36m$\\
\hline
\end{tabular}
\caption{The table shows the description and statistics of the datasets. We generate synthetic actions via IC model for CA, EE and CY datasets. The number of tuples (and thus actions) are varied for different experiments.\label{table:data_description}}
 \end{table}

\begin{figure}[t]
    \centering
    \subfloat[CA (varying k)]{\includegraphics[width=0.23\textwidth]{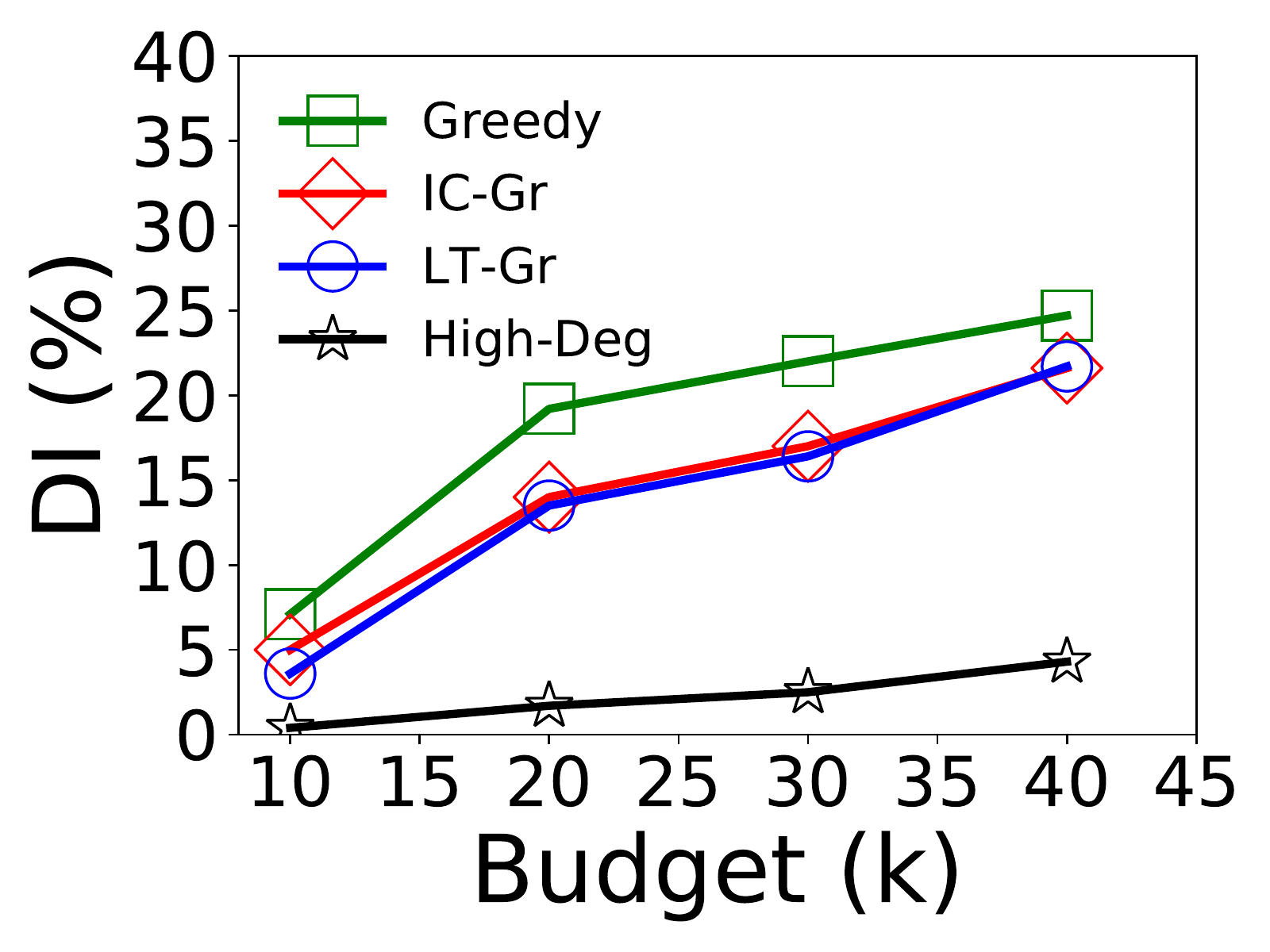}\label{fig:baseline_syn}}
    \subfloat[CA (varying |X|)]{\includegraphics[width=0.23\textwidth]{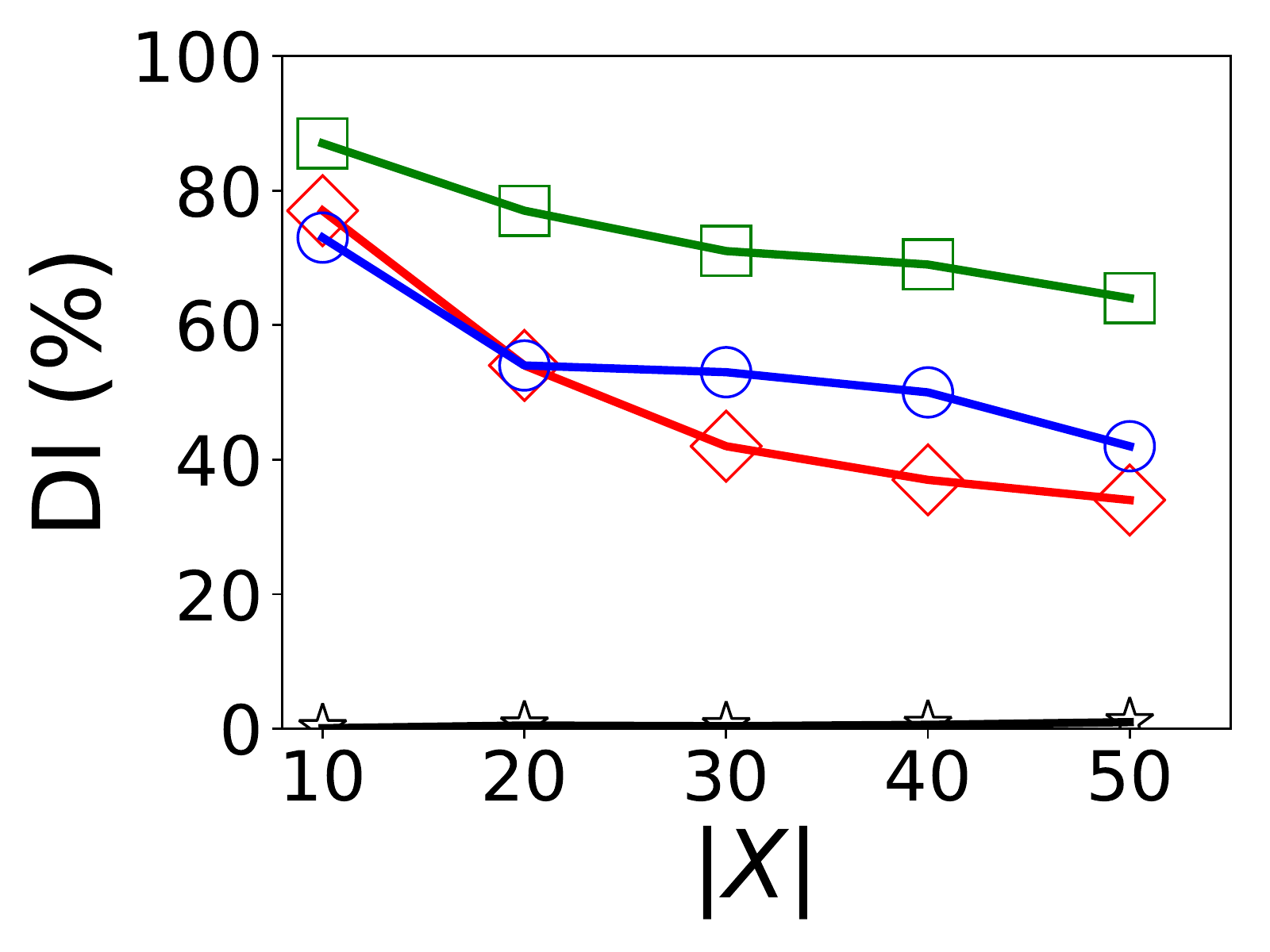}\label{fig:source_synthetic_sil}} \\
    \subfloat[FXS (varying k)]{\includegraphics[width=0.23\textwidth]{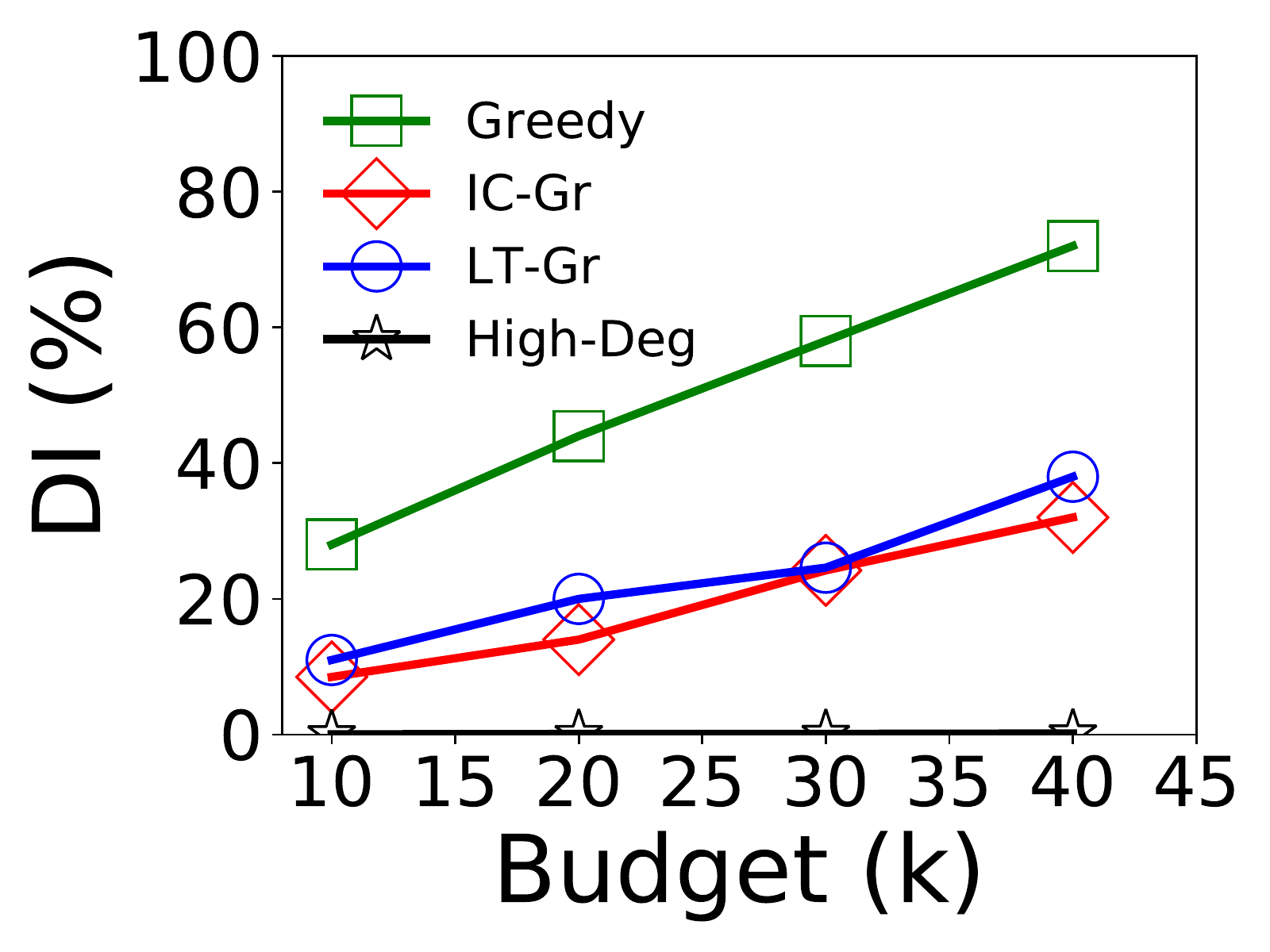}\label{fig:baseline_flix}} 
    \subfloat[FXS (varying |X|)]{\includegraphics[width=0.23\textwidth]{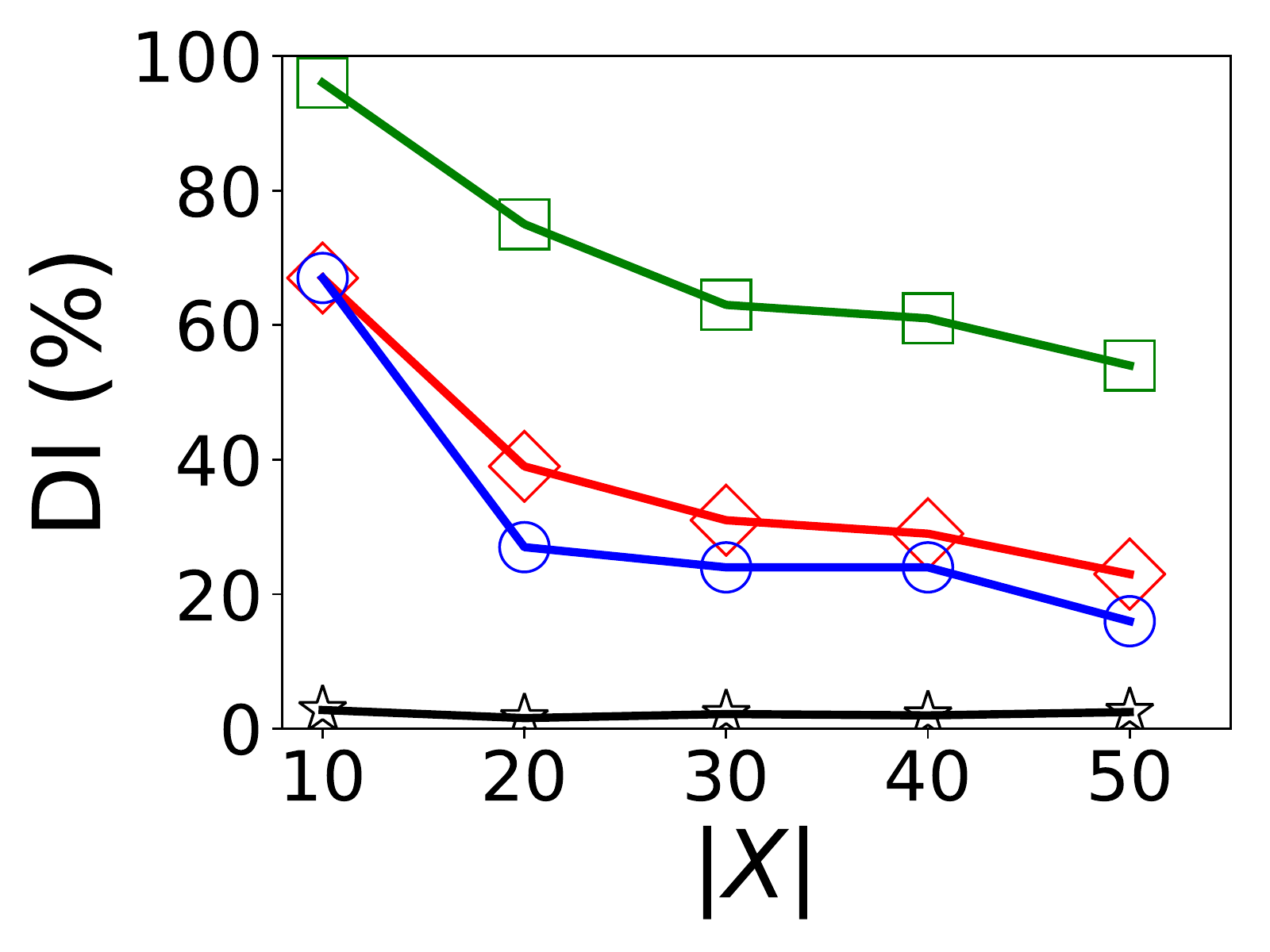}\label{fig:source_flix_sil}}\\
    \subfloat[FCS (varying k)]{\includegraphics[width=0.23\textwidth]{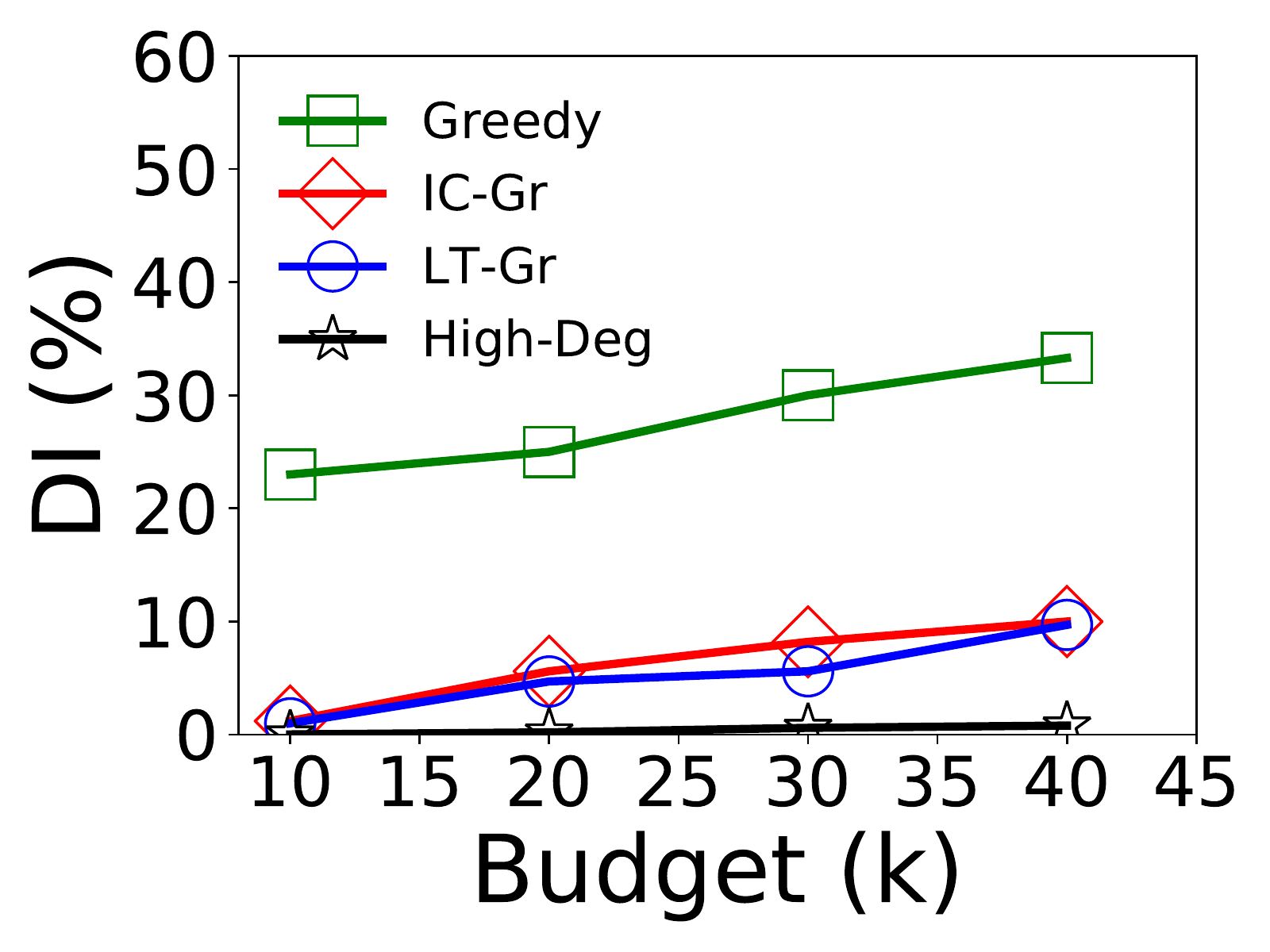}\label{fig:baseline_flick}}
    \subfloat[FCS (varying |X|)]{\includegraphics[width=0.23\textwidth]{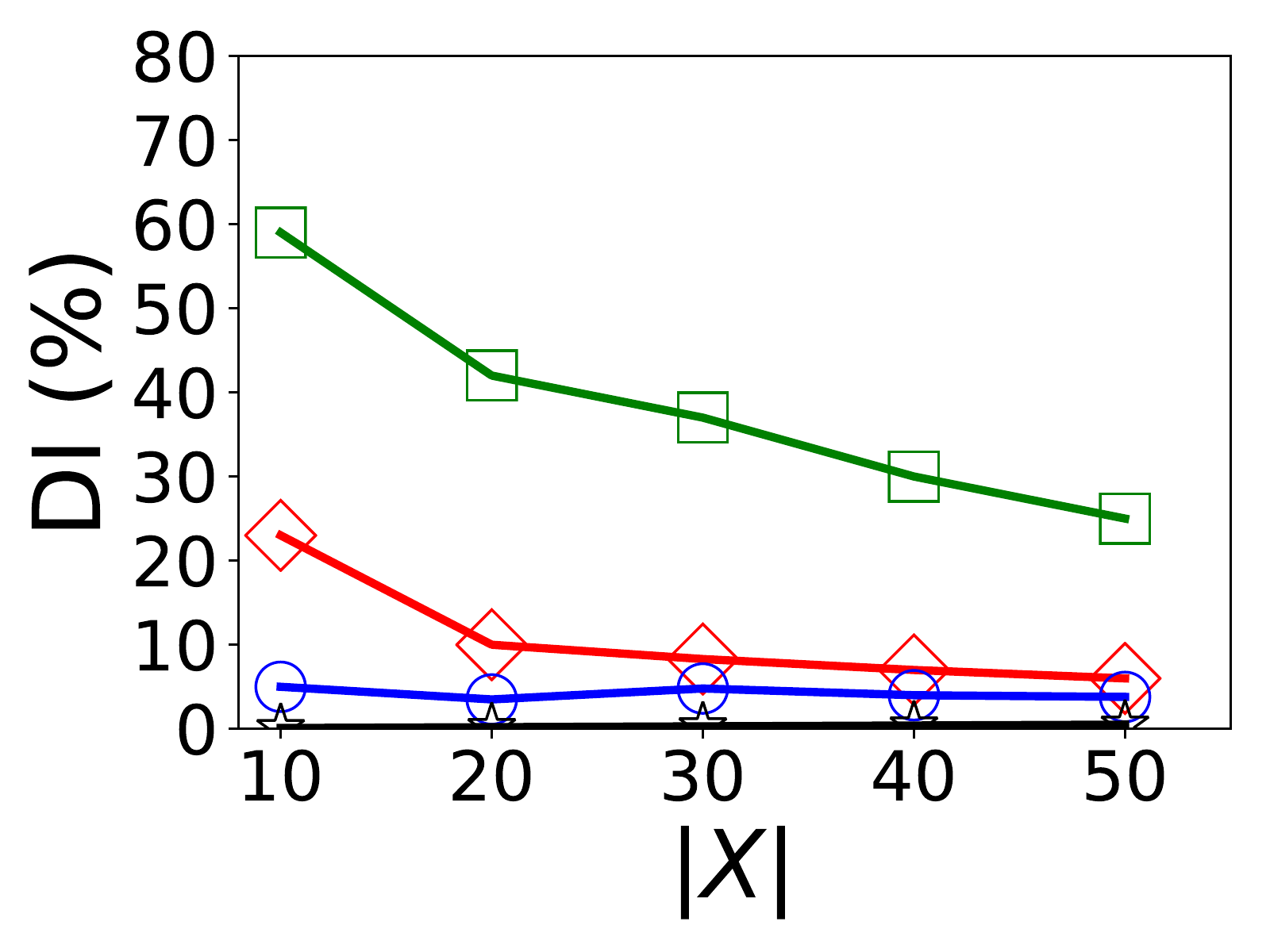}\label{fig:source_flick_sil}}
    \caption{[BIL] (a, c, e) Decrease in Influence (DI) produced by different algorithms. Greedy outperforms the baselines by up to $40\%$. (b, d, f) DI produced by different algorithms varying the size of the target set, $X$ with $k=30$. \label{fig:SIL_source_baseline}}
 
\end{figure}

We evaluate the quality and scalability of our algorithms using synthetic and real networks. Solutions were implemented in Java and experiments conducted on $3.30$GHz Intel cores with $30$ GB RAM. 

\noindent
\textbf{Datasets: } The datasets used in the experiments are the following:
1) \textbf{Flixster \cite{goyal2011}}: Flixster is an unweighted directed social
graph, along with the log of performed actions. The log has triples of $(u, a, t)$ where user $u$ has performed action $a$ at time $t$. Here, an action
for a user is rating a movie.
2) \textbf{Flickr \cite{mislove-2007-socialnetworks}}: This is a photo sharing platform. Here, an action would be joining an
interest group. 
3) \textbf{Synthetic}: We use the structure of real datasets that come from different genre (e.g., co-authorship, social). The networks are available online\footnote{https://snap.stanford.edu}.  We synthetically generate the actions and create associated tuples. Synthetic actions are generated assuming the Independent Cascade (IC) \cite{kempe2003maximizing} model. The ``ca-AstroPh" dataset is a Collaboration network of Arxiv Astro Physics. 
In the ``Youtube" social network, users form friendship with others and can create groups which other users can join.  Table \ref{table:data_description} shows the statistics of the datasets. We use the small extracted networks (from Flixster and Flickr) for the quality-related experiments as our baselines are not scalable. To show scalability of our methods, we extract networks of different sizes from the raw large Flixster and Flickr data. For all the networks, we learn the influence probabilities via the widely used method proposed by Goyal et al. \cite{goyal2010learning}.

\noindent 


\textbf{Performance Metric:} The quality of a solution set $B$ (a set of edges) is the percentage of decrease in the influence of the target set $X$. Thus, the \textit{Decrease in Influence (DI)} in percentage is:
\begin{equation}
\label{eq:quality_metric}
DI(B)= \frac{(\sigma_{cd}(G,X)-\sigma_{cd}(G^m,X))}{\sigma_{cd}(G,X)} \times 100
\end{equation}

\textbf{Other Settings: } The set of target nodes $X$ is randomly selected from the set of top $150$ nodes with highest number of actions. We build the candidate set $C$ with those edges that appear at least once in any action graph. The number of Monte Carlo simulations for IC and LT-based baselines is at least $1000$ if not specified otherwise.\\ 

\subsection{Experiments: BIL}

 \begin{table}[t]
\centering

\begin{tabular}{| c | c | c | c |}
\hline
&\multicolumn{3}{c|}{\textbf{FXS:} $\#$ (tuples, actions)$\times 10^3$}  \\
\hline

\textbf{Budget}& ($30,1.7$) & ($50,4.8$)& ($75,6.9$)\\
\hline
$k=50$ & $58$& $61$&  $68$ \\
\hline
$k=75$ &   $73$&$83$& $85$\\
\hline
$k=100$ &  $85$&$88$ & $91$\\
\hline
&\multicolumn{3}{c|}{\textbf{FCS:} $\#$ (tuples, actions)$\times 10^3$}  \\
\hline
&($20,2.6$) & ($30,3.8$) &($50,5.8$)\\
\hline
$k=50$ & $208$ &$383$ &$1187$\\
\hline
$k=75$ &  $269$ &$579$ &$1891$\\
\hline
$k=100$ &  $356$ &$780$ &$2551$\\
\hline
\end{tabular}

\caption{\textbf{[BIL] }Running Times (Scalability) of Greedy varying number of tuples. The times are in seconds. The number of tuples and actions are in thousands. \label{table:scalability_tuple}}

 \end{table}
 
 \textbf{Baselines:} We consider three baselines in these experiments: 1) \textbf{IC-Gr \cite{kimura2008minimizing}: } Finds the top $k$ edges based on the greedy algorithm proposed in \cite{kimura2008minimizing}, which minimizes influence via edge deletion under the IC model. 2) \textbf{LT-Gr \cite{Khalil2014}: } Finds the top $k$ edges based on the greedy algorithm proposed in \cite{Khalil2014}. Here, the authors minimize the influence of a set of nodes according to the LT model via edge deletion.  Note that we also apply optimization techniques proposed in \cite{Khalil2014} for both of these baselines. 3) \textbf{High-Degree:} This baseline selects edges between the target nodes $X$ and the top-$k$ high degree nodes. We have also applied the selection of top edges uniformly at random and using the Friends of a Friend (FoF) algorithm. The results are not significantly different (within $1\%$) from High-Deg. Thus, we use High-Deg as the representative baseline for them. 
 
 \noindent
 \subsubsection{Quality (vs Baselines)}
 \label{sec: qual_sil}
 
 We compare our Greedy algorithm (Algorithm \ref{algo:Greedy}) against the baseline methods on three datasets: CA, FXS, and FCS. The target set size is set as $30$.
Figures \ref{fig:baseline_syn}, \ref{fig:baseline_flix}, \ref{fig:baseline_flick} show the results, where the measure for quality is DI($\%$) (Eq. \ref{eq:quality_metric}). Greedy takes a few seconds to run and significantly outperforms the baselines (by up to $40\%$). The running time of Greedy is much lower as it avoids expensive Monte-Carlo simulations. For CA, the action graphs are generated through IC model. Therefore, the baseline IC-Gr produces better results on CA than other two datasets.  

\begin{figure}[t]
    \centering
    \subfloat[Quality on CA]{\includegraphics[width=0.22\textwidth]{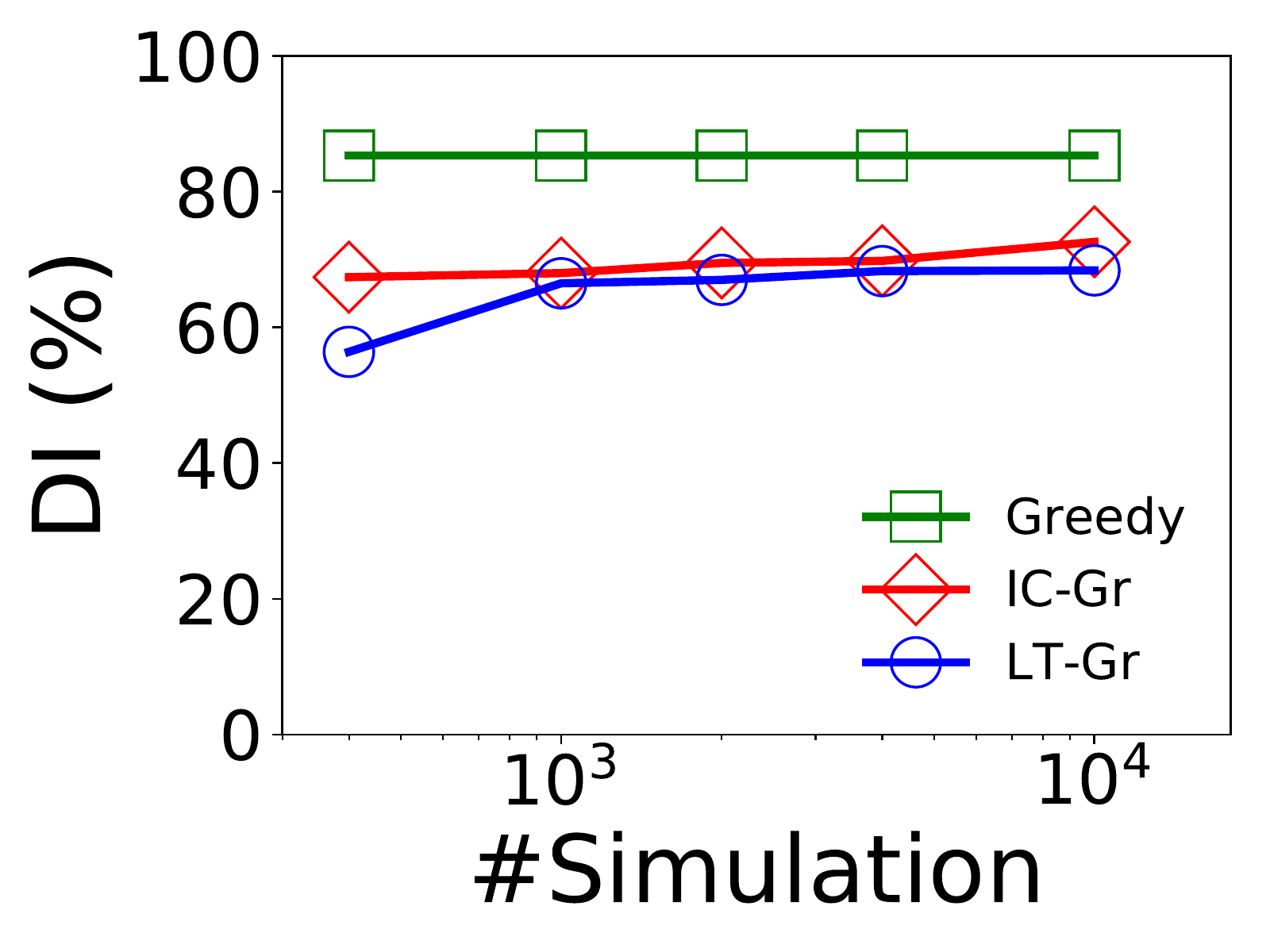}\label{fig:sim_qual_syn}}
    \subfloat[Quality on FXS]{\includegraphics[width=0.22\textwidth]{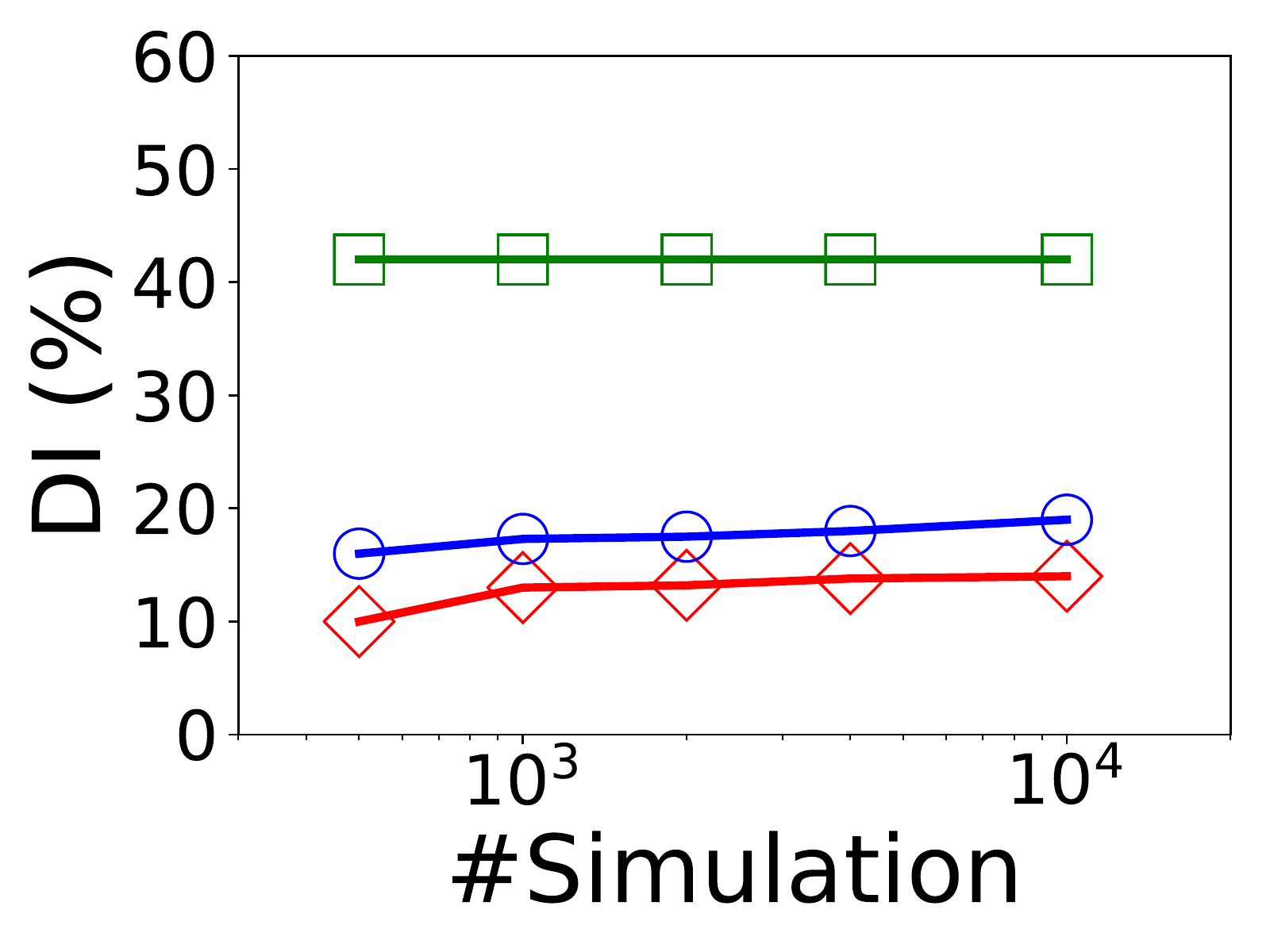}\label{fig:sim_qual_flix}}\\
    \subfloat[Time on CA]{\includegraphics[width=0.22\textwidth]{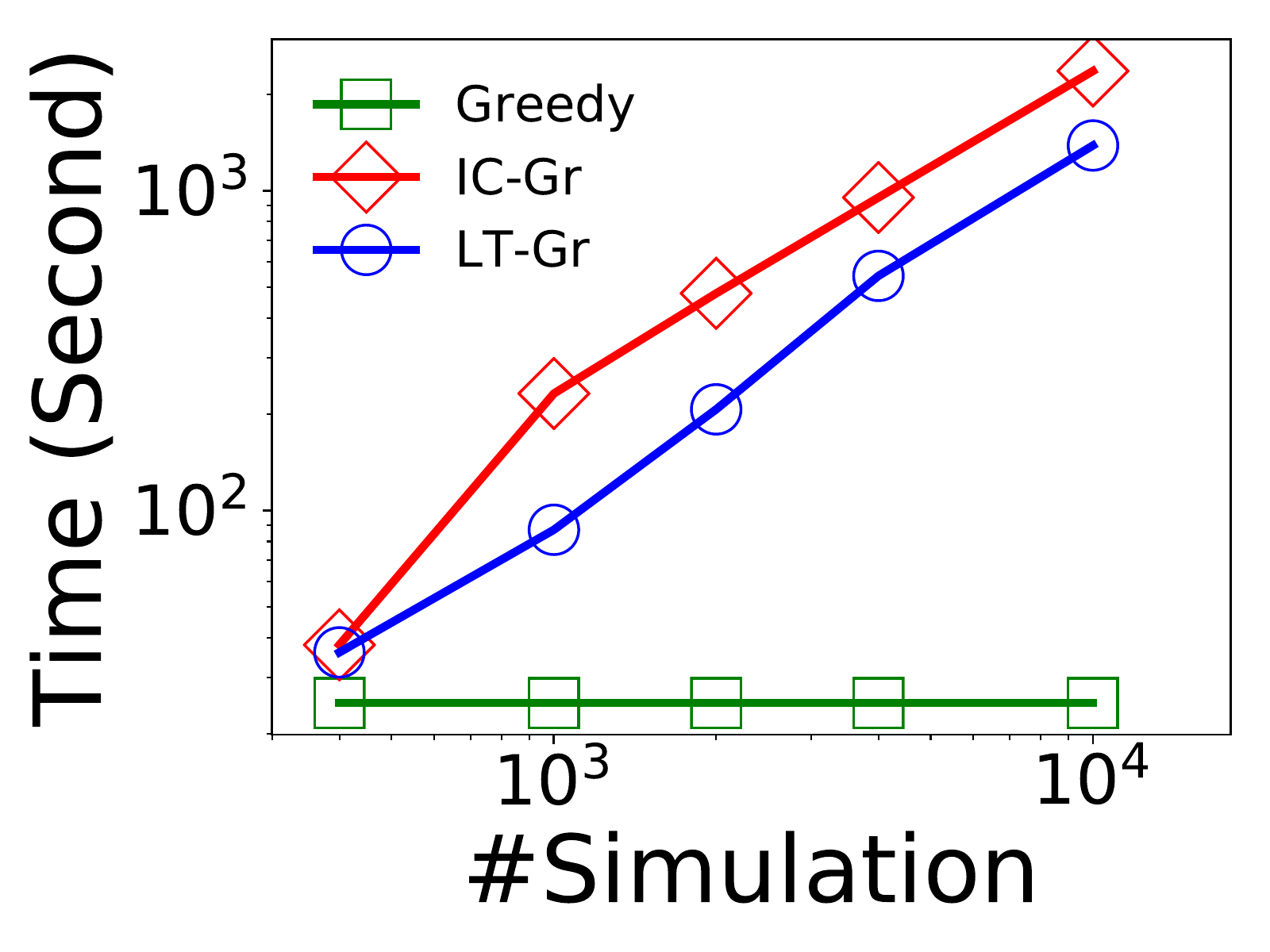}\label{fig:sim_time_syn}}
    \subfloat[Time on FXS]{\includegraphics[width=0.22\textwidth]{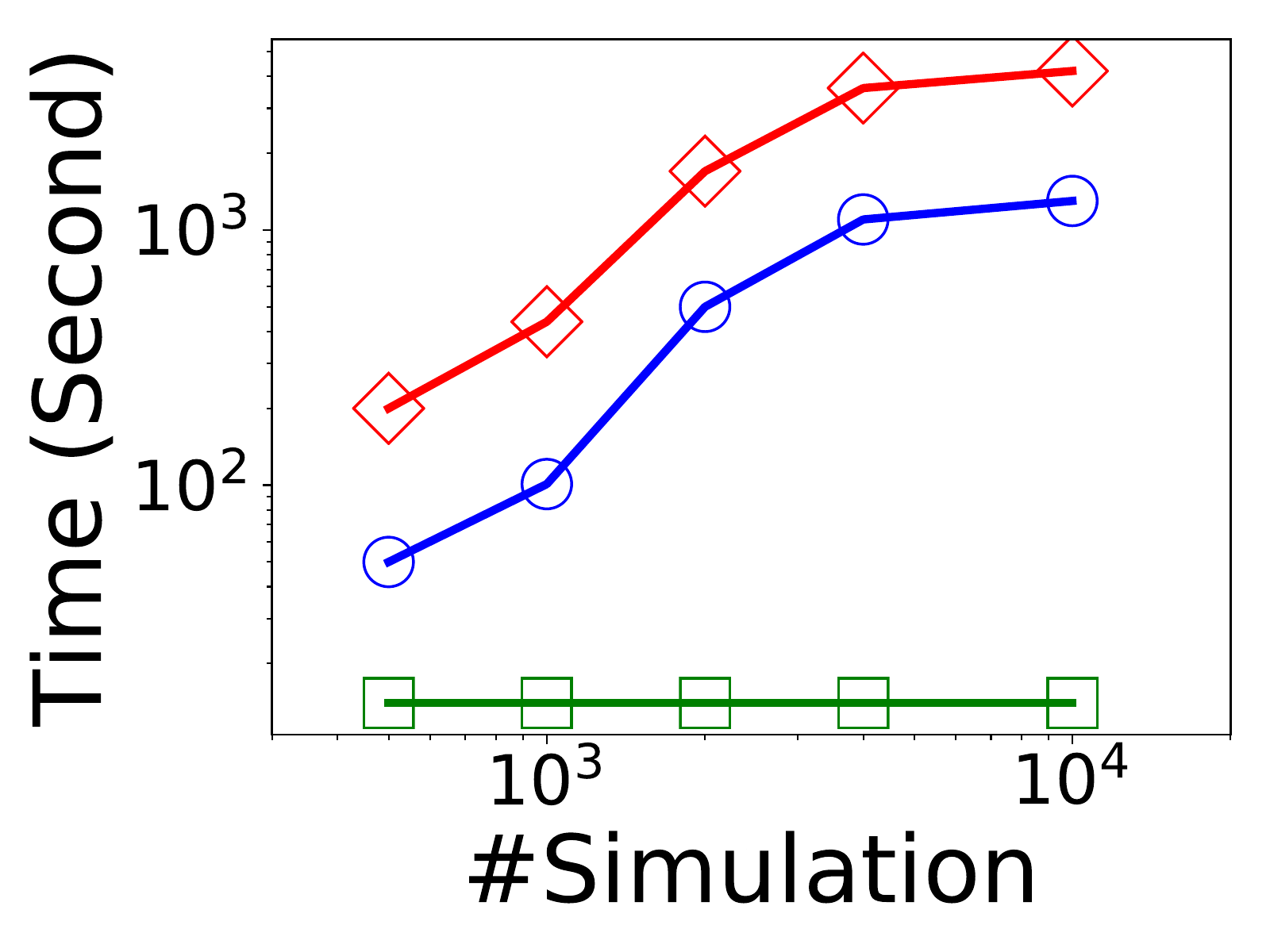}\label{fig:sim_time_flix}}
    \caption{[BIL] Comparison of our greedy algorithm and simulation based baselines varying number of simulations: (a-b) Quality on CA, FXS and (c-d) Running times on CA, FXS. \label{fig:sim_param}}
\end{figure}

\noindent
 \subsubsection{Scalability of Greedy}
 \label{sec: scale_sil}
We show the scalability of our Greedy algorithm (Algorithm \ref{algo:Greedy}) by increasing the number of tuples (thus the number of actions) as well as the size of the graph. Table \ref{table:scalability_tuple} shows the results on two datasets, FXS and FCS. As FCS is a graph with higher density than FXS, the number of tuples has higher effect on the running time in FCS. Note that we consider all the edges that appear in one of the actions in our candidate set of edges. A larger candidate set results in longer running time. However our algorithm only takes around $2$ and $43$ minutes to run for $75k$ and $50k$ tuples in FXS and FCS, respectively.

Table \ref{table:scalability_graph_BIL} shows the results varying the graph size. The running times are dominated by the size of both the graphs and the candidate sets. Greedy takes approximately 16 minutes on CY with 1m nodes and 6k candidate edges, whereas, it takes 67 minutes on FX with 200k nodes and 51k candidate edges.
 
 \begin{table}[t]
\centering

 \begin{tabular}{| c | c | c | c |c|c|}
\hline
\textbf{Dataset}& $|V|$ & Actions & Tuples & $|C|$ & Time (sec)\\
\hline
EE & 265k & 5k &  326k & 4.1k & 637 \\
\hline
CY &  1.1m & 5k &  313k & 6.3k & 950 \\
\hline
FX &  200k & 2.6k &  200k & 51k & 4020 \\
\hline
\end{tabular}
 
 \caption{\textbf{[BIL]} Running Times (Scalability) of Greedy varying graph size for $|X|=30$ and $k=30$.  \label{table:scalability_graph_BIL}}

 \end{table}

 \subsubsection{Parameter variations} \label{sec:param_BIL}
 We also analyze the impact of varying the parameters. We explain the effect of varying budget, number of tuples, and size of the graph over the performance of the algorithms in Sections \ref{sec: qual_sil} and \ref{sec: scale_sil}. Here we assess the impact of the number of target nodes (size of the target set, $|X|$). We also vary the number of simulations for LT-Gr and IC-Gr.  
 
First we vary the size of the target set $X$. Figures \ref{fig:source_synthetic_sil}, \ref{fig:source_flix_sil} and \ref{fig:source_flick_sil} show the results for CA, FXS and FCS, respectively. We fix the budget $k=30$ for these experiments. Greedy provides better $DI$ across all the target sizes and the datasets. With the increase in target set size, DI decreases for the top three algorithms. A larger target size would have a higher influence to reduce. Thus, with the same number of edges removed, the DI would decrease for larger target set. Also, DI is lower for FCS as it is much denser than CA and FXS.

We also evaluate how LT-Gr and IC-Gr are affected by the number of simulations. 
 We fix the target set size, $|X|=30$ and the budget, $k=20$. Figure \ref{fig:sim_param} shows the results. Our algorithm produces better results even when the baselines perform $10k$ simulations. By comparing figures \ref{fig:sim_qual_syn} and \ref{fig:sim_qual_flix}, it is evident that the baseline IC-Gr performs better than LT-Gr in CA as the synthetic actions are generated via IC model. So, intuitively, IC based greedy algorithm, IC-Gr should perform better than LT-Gr. Figures \ref{fig:sim_time_syn} and \ref{fig:sim_time_flix} also show that our method is $1-4$ orders faster than the simulation based baselines.

\subsection{Experiments: ILM}
\label{sec:ilm_exp}

\begin{figure}[t]
    \centering
    \subfloat[CA (b=1) ]{\includegraphics[width=0.21\textwidth]{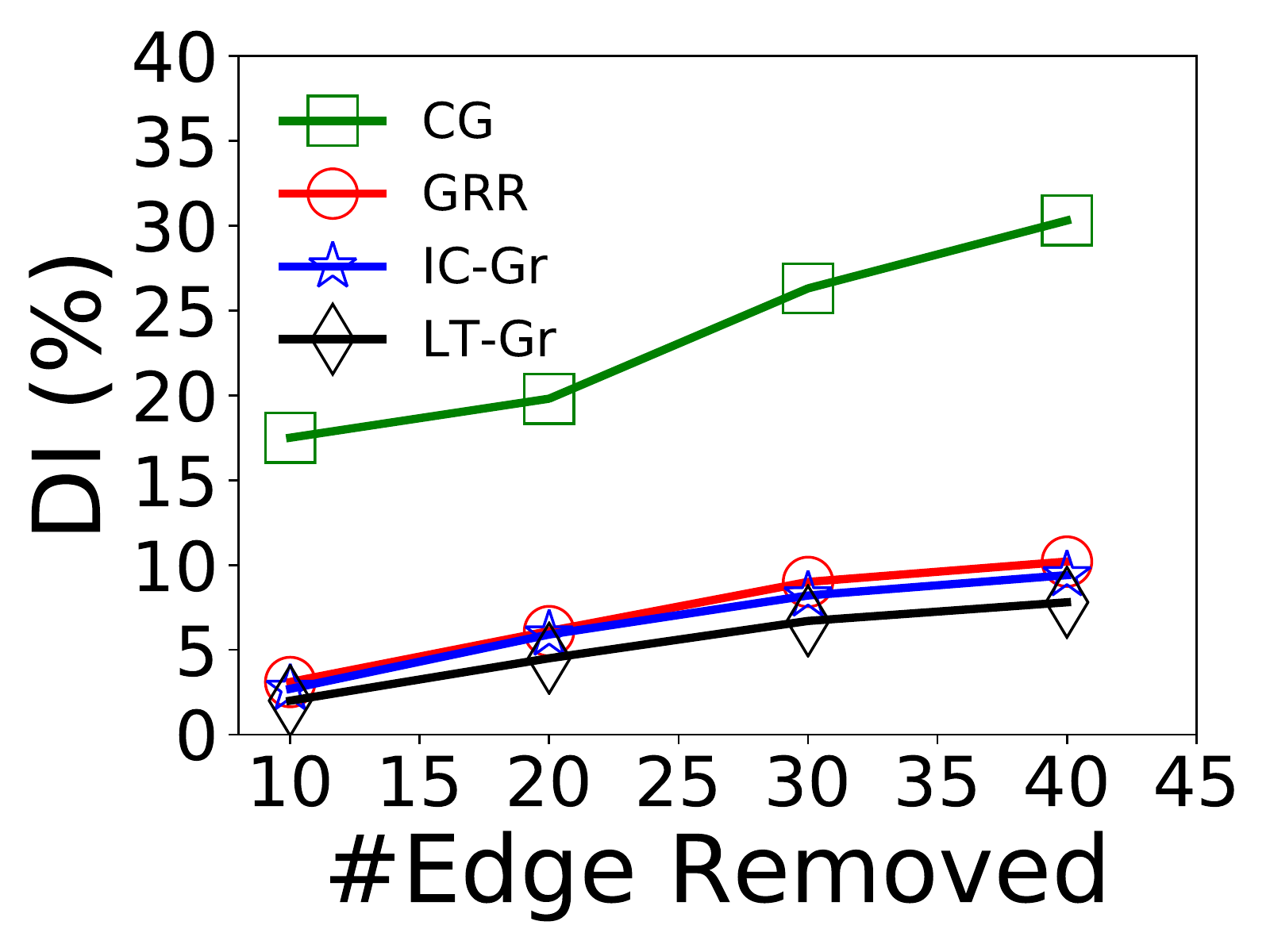}\label{fig:baseline_syn_ILM1}}
    \subfloat[CA (b=2)]{\includegraphics[width=0.21\textwidth]{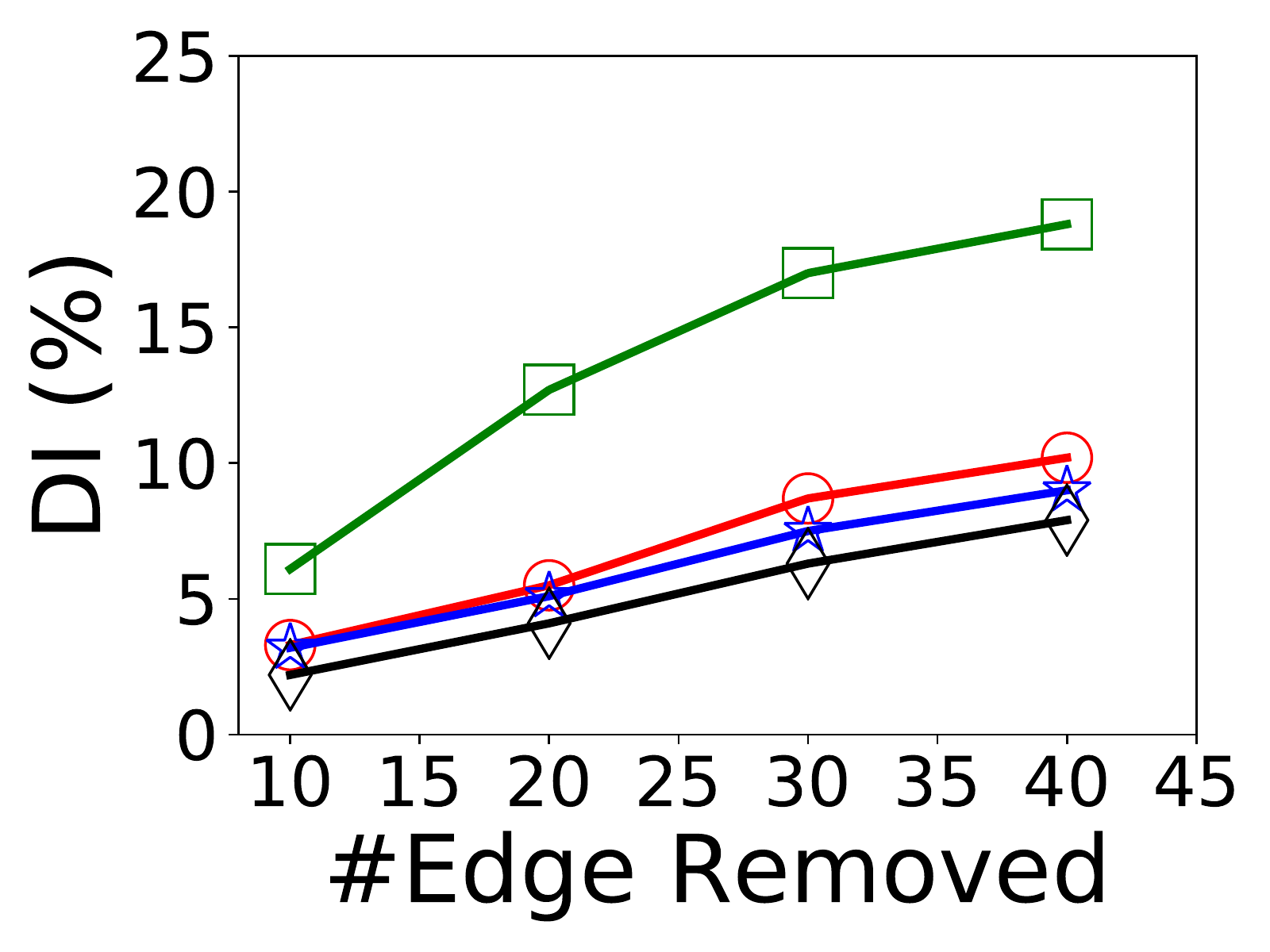}\label{fig:baseline_syn_ILM2}} \\
    
    
    \subfloat[FCS (b=1)]{\includegraphics[width=0.21\textwidth]{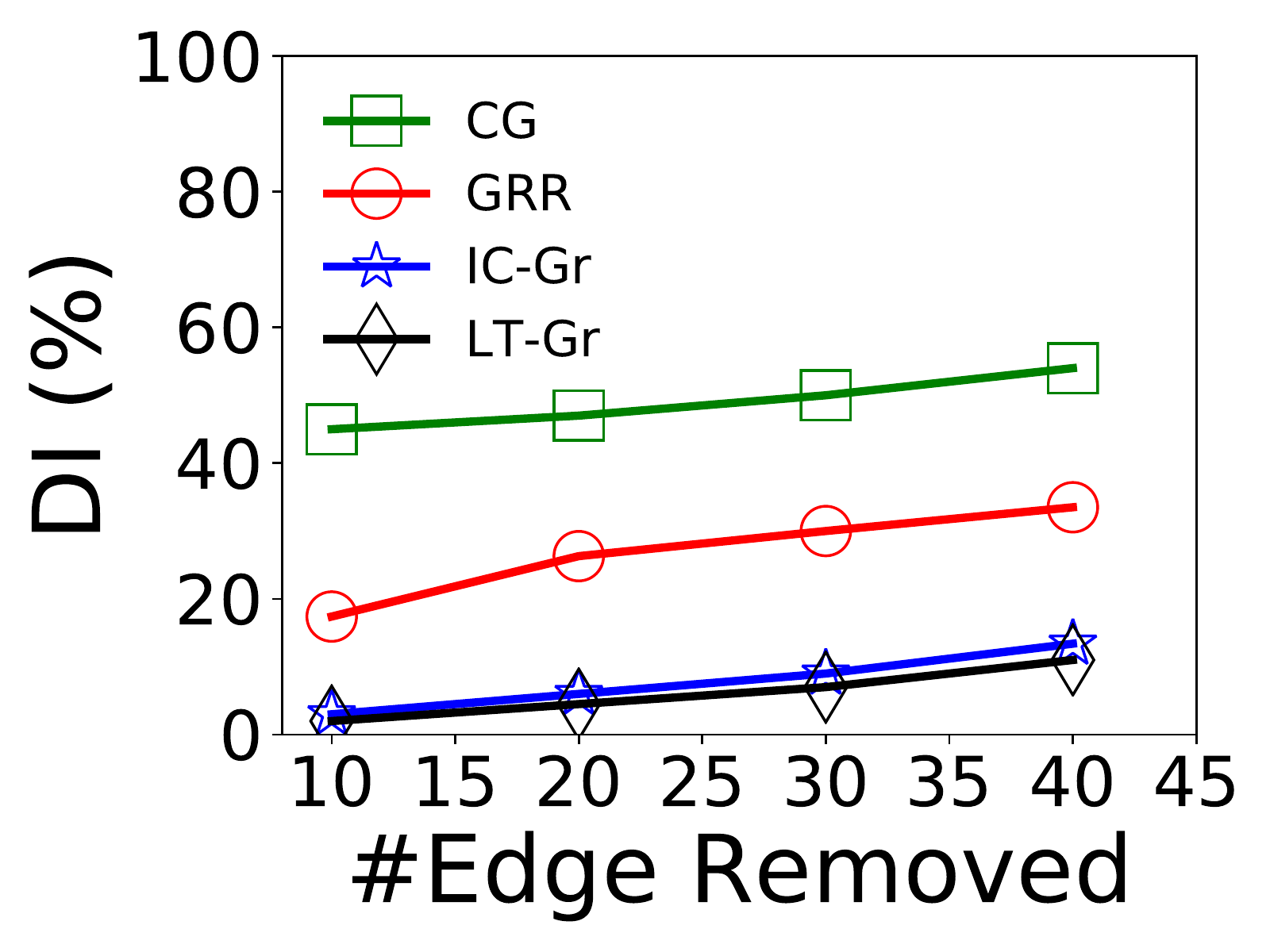}\label{fig:baseline_flick_ILM1}}
    \subfloat[FCS (b=2)]{\includegraphics[width=0.21\textwidth]{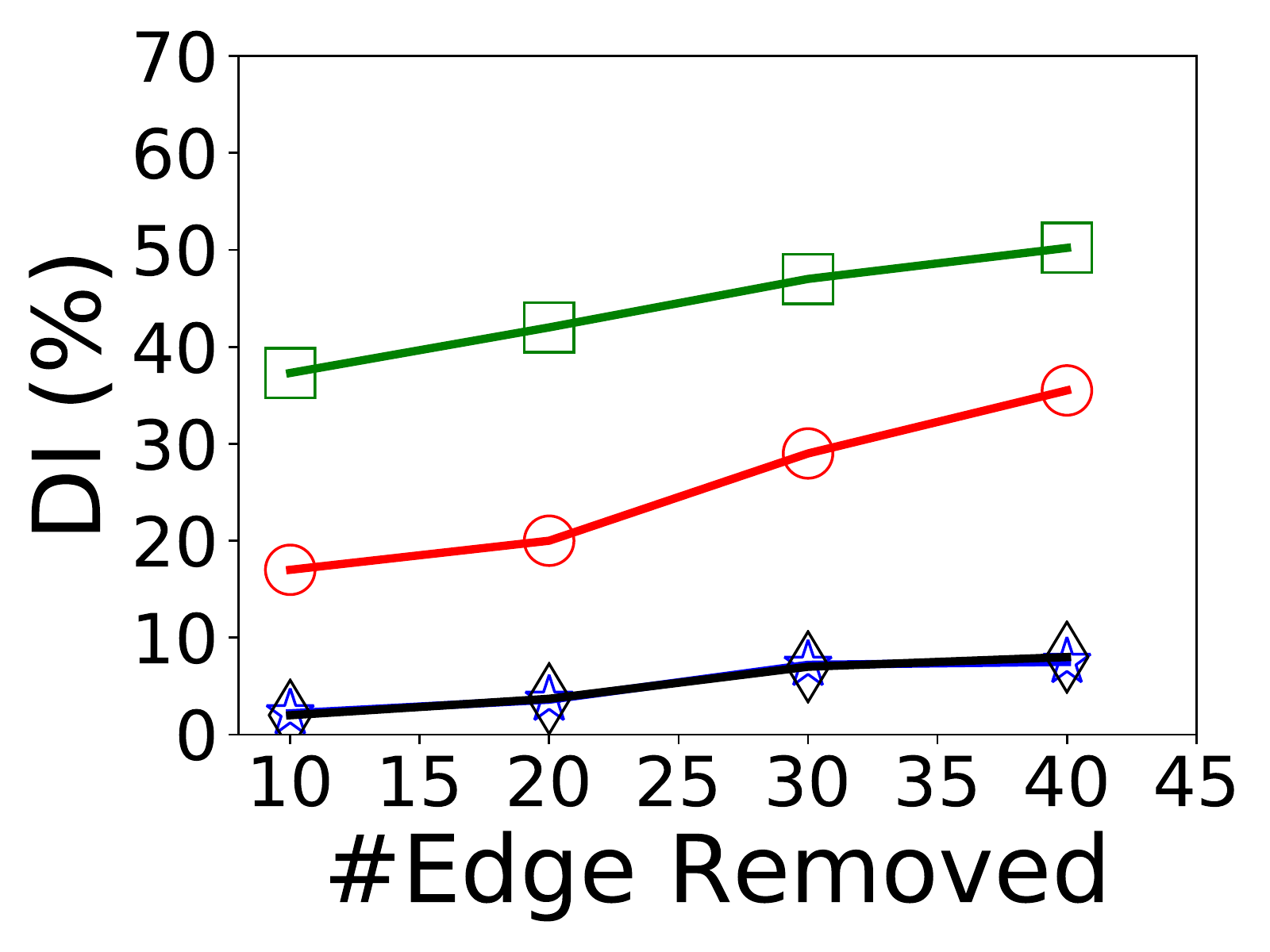}\label{fig:baseline_flick_ILM2}}
    \caption{ \textbf{[ILM] } Decrease in Influence for $b=1$ and $b=2$ produced by different algorithms on (a-b) CA and (c-d) FCS. Our algorithm, CG outperforms the baselines by up to $20\%$. \label{fig:vs_baselines_ILM1_ILM2}}
\end{figure}

\textbf{Baselines and other settings:}  To compare with our \textit{Continuous Greedy (CG)} algorithm we consider three baselines in these experiments: (1) \textbf{Greedy with Restriction (GRR): } Finds the feasible edges based on the greedy algorithm proposed for BIL. The greedy algorithm chooses the best ``feasible" edge that respects the constraint of maximum ($b$) edges removed. (2-3) We also apply \textbf{IC-Gr} and \textbf{LT-Gr} with the edge removal constraint for each node. The number of samples and iterations used in CG are $s=20$ and $\tau = 100$, respectively. After obtaining the solution vector from CG, we run randomized rounding for $50$ times and choose the best solution.

 \subsubsection{Quality (vs Baselines)} We compare the Continuous Greedy (CG) algorithm against the baseline methods on FCS and CA (FXS is omitted due to space constraints). The target set size is set as $30$. We experiment with $b=1$ and $b=2$. Figure \ref{fig:vs_baselines_ILM1_ILM2} shows the results. CG significantly outperforms the baselines by up to $20\%$. GRR does not produce good results as it has to select the feasible edge that does not violate the maximum edge removal constraint $b$. While maintaining feasibility, GRR cannot select the current true best edge. 

 \begin{table}[t]
\centering

\begin{tabular}{| c | c | c | c | }
\hline
&\multicolumn{3}{c|}{\textbf{FXS:} $\#$ (tuples, actions)$\times 10^3$}  \\
\hline

\textbf{$\#$Edge Removed}& ($30,1.7$) & ($50,4.8$)& ($75,6.9$)\\
\hline
$20$ & $7.5$& $20.7$&  $69.4$ \\
\hline
$40$ &   $7.1$ &$16.8$& $69.4$\\
\hline
$60$ &  $7.4$&$16.7$ & $69.5$\\
\hline
&\multicolumn{3}{c|}{\textbf{FCS:} $\#$ (tuples, actions)$\times 10^3$} \\
\hline
& ($20,2.6$) & ($30,3.8$) &($50,5.8$)\\
\hline
$20$ & $28.1$ &$64.7$ &$180$\\
\hline
$40$ &  $29.1$ &$64.6$ &$181$\\
\hline
$60$ &  $29.2$ &$63.1$ &$167$\\
\hline
\end{tabular}

\caption{\textbf{[ILM] }Running Times (Scalability) of CG varying number of tuples for $|X|=20$ and $b=2$. The running times are in minutes. The number of tuples and actions are in thousands. \label{table:scalability_tuple_ILM}}

 \end{table}
 
 \begin{figure}[t]
    \centering
    \subfloat[Quality on CA]{\includegraphics[width=0.22\textwidth]{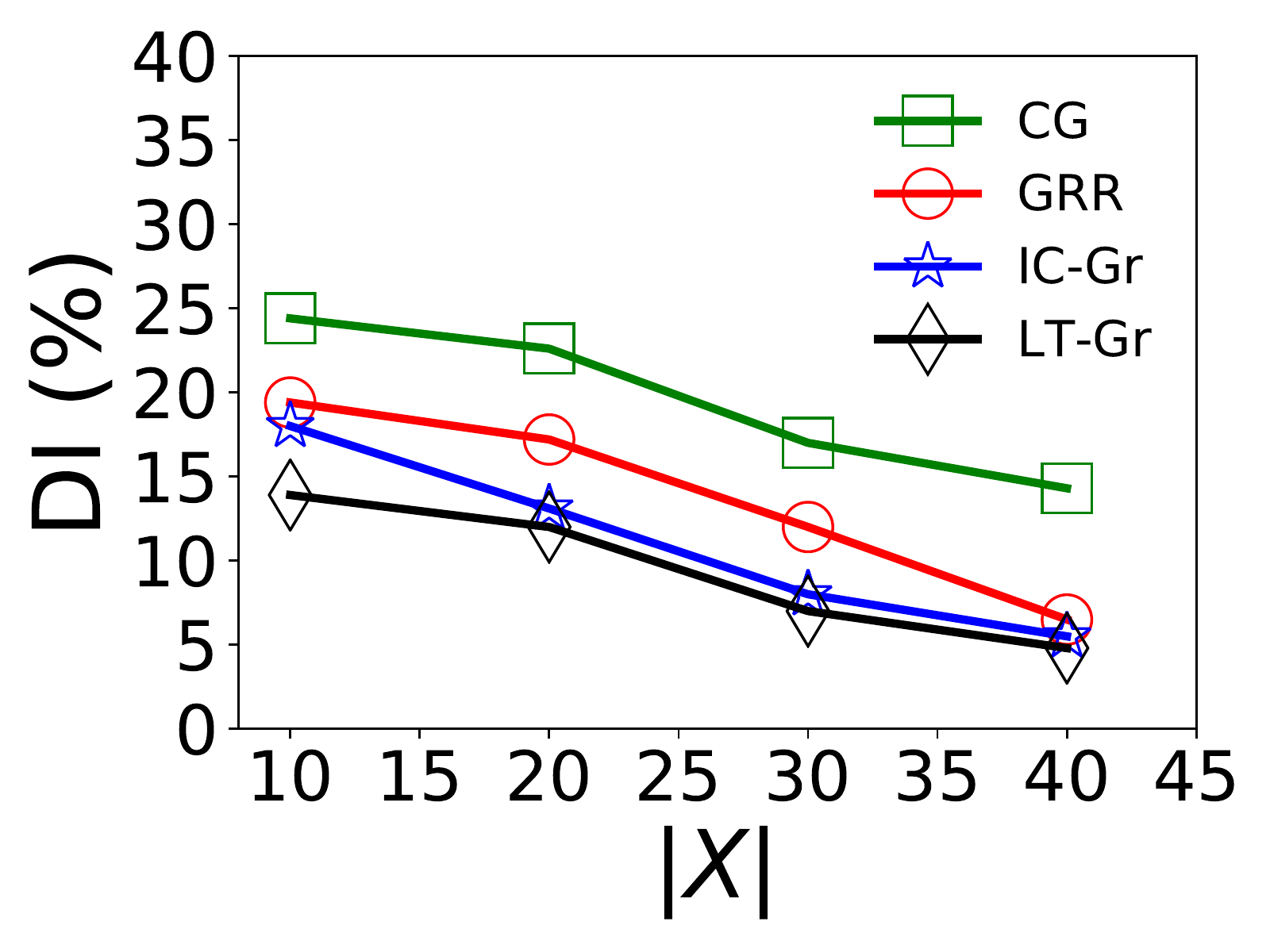}\label{fig:source_synthetic_ILM}}
    \subfloat[Quality on FXS]{\includegraphics[width=0.22\textwidth]{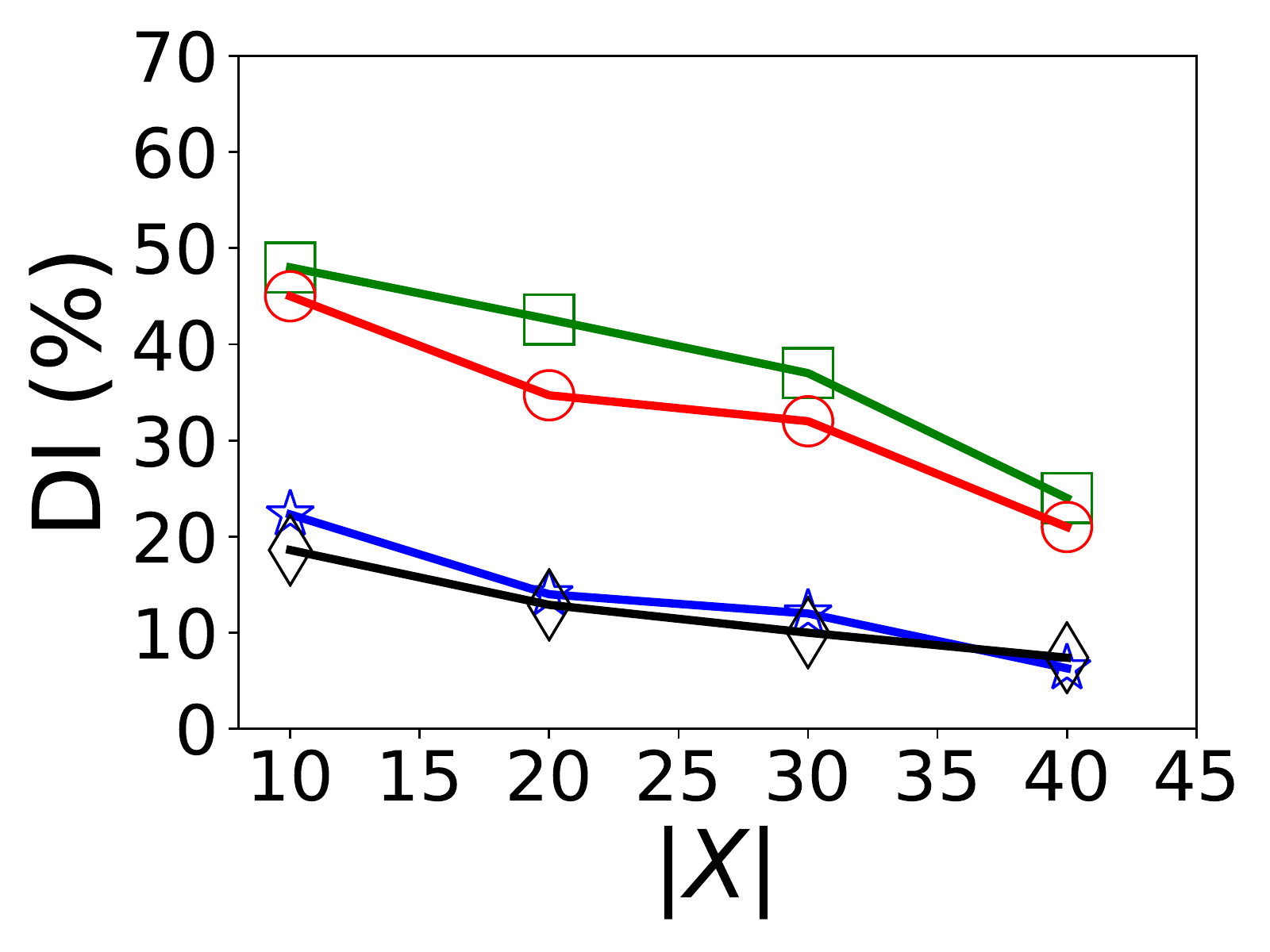}\label{fig:source_flix_ILM}}
    \caption{\textbf{[ILM]} Decrease in Influence (DI) produced by different algorithms varying the size of the target set, $X$ when $b=2$. \label{fig:ILM_source}}
\end{figure}
 
 
 \begin{table}[t]
\centering

\begin{tabular}{| c | c | c | c | c|c|c|}
\hline
&\multicolumn{6}{c|}{\textbf{FCS:} $\#$ (tuples, actions)$\times 10^3$}  \\
\hline

\textbf{$\#$Edge Removed}& \multicolumn{2}{c|}{ ($20,2.6$)} & \multicolumn{2}{c|}{($30,3.8$)}& \multicolumn{2}{c|} {($50,5.8$)}\\
\hline
 & \textbf{CG}&\textbf{GRR} & \textbf{CG} & \textbf{GRR} & \textbf{CG} &  \textbf{GRR} \\
\hline
$20$ & $33$& $22$ &  $50$ & $41$ & $35$ &  $20$ \\
\hline
$40$ &   $42$ &$32$& $53$ &   $44$ &$45$& $35$\\
\hline
$60$ &  $44$&$35$ & $61$ &  $54$&$54$ & $40$\\
\hline
\end{tabular}

\caption{\textbf{[ILM] } Decrease in Influence ($\%$) in FCS by Continuous Greedy (CG) vs GRR varying the number of tuples. The number of tuples and actions are in thousands. \label{table:qual_scalability_tuple_ILM_FCS}}

 \end{table}

 \noindent
 \subsubsection{Scalability of Continuous Greedy.}
 CG (Algorithm \ref{algo:CG}) is generally slower than GRR. We evaluate the running time of CG while increasing the number of tuples (thus, the number of actions). Table \ref{table:scalability_tuple_ILM} shows the results on two datasets, FXS and FCS. Because of higher density and thus larger candidate set, CG takes longer in FCS. Furthermore, the increment in budget does not affect the running time for CG. These observations validate the running time analysis for CG (Section \ref{sec:continuous_relaxation}). We have also shown the quality in DI ($\%$) produced by CG and GRR (other baselines are not scalable). Table \ref{table:qual_scalability_tuple_ILM_FCS} shows the results on FCS data (the results for FXS are in the Appendix. CG outperforms GRR by up to $15\%$. Other scalability results varying graph size are in the Appendix.

\subsubsection{Parameter Variation}
Finally, we analyze the impact of the variation of the parameter $X$ (i.e., the size of the target set) over CG. We have considered the effect of varying budget (along with $b$) and number of tuples earlier. 
The size of the target set $X$ is varied and we observe its effect in Figure \ref{fig:ILM_source}. We set $b=2$, and remove $20$ edges for these experiments. CG provides better $DI$ consistently across target sizes and datasets (the results using FCS have similar trend and are omitted here). With the increase of target set size, DI generally decreases for all the algorithms. A larger target size would have a higher influence to be reduced. Thus, with the same number of edges removed, the DI would decrease for a larger target set.

\section{Previous Work}
\textbf{Boosting and controlling propagation:} 
The influence boosting or limitation problems via network modifications are orthogonal to the classical influence maximization task \cite{kempe2003maximizing}. In these modification problems, the objective is to optimize (maximize or minimize) the content spread via structural or attribute-level change in the network. Previous work has also addressed the influence limitation problem in the SIR model \cite{Tong2012GML,Gao2011,schneider2011suppressing}. The objective is to optimize specific network properties in order to boost or contain the content/virus spread. For instance, Tong et al. proposed methods to add (delete) edges to maximize (minimize) the eigenvalue of the adjacency matrix. 

The influence spread optimization problem has been studied under the IC model via network design \cite{kimura2008minimizing, bogunovic2012robust,sheldon2012maximizing, chaoji2012recommendations,lin2017boosting} and injecting an opposite campaign \cite{budak2011limiting,nguyen2012containment}. We mainly focus on the network design problem here. Bogunovic \cite{bogunovic2012robust} addressed the minimization problem via node deletion. On the other hand, Sheldon et al. \cite{sheldon2012maximizing} studied the node addition problem and proposed expensive algorithms based on mixed integer programming. Kimura et al. \cite{kimura2008minimizing} proposed greedy algorithms for the same. While Chaoji et al. \cite{chaoji2012recommendations} studied the problem of boosting the content spread via edge addition, Lin et al. \cite{lin2017boosting} investigated the same via influencing initially uninfluenced users. 

Boosting and controlling the influence via edge addition and deletion,  respectively, were also studied under the Linear Threshold (LT) model by Khalil et al. \cite{Khalil2014}. They showed the supermodular property for the objective functions and then applied known approximation guarantees. The influence minimization problem was also studied under a few variants of LT model. \cite{kuhlman2013blocking,he2012influence,chen2013information}. 
In summary, the approaches for optimizing influence (propagation) are mostly based on the well-known diffusion models such as SIR, LT and IC. However, our work addresses the influence minimization problem based on available cascade information.

\textbf{Optimization over matroids:} Matroids have been quite popular for modelling combinatorial problems \cite{nemhauser1978, chekuri2004maximum}. Nemhauser \cite{nemhauser1978} introduced a few optimization problems under matroids. Vondrak \cite{Vondrak08} addressed matroid optimization with a continuous greedy technique for submodular functions. Calinescu et al. \cite{calinescu2011maximizing} and Chekuri et al. \cite{chekuri2010dependent} proposed rounding techniques for continuous relaxation of submodular functions under matroids. 

\textbf{Other network modification problems:} 
We also provide a few details about previous work on other network modification (design) problems. A set of design problems were introduced in \cite{paik1995}. Lin et al.~\cite{lin2015} addressed a shortest path optimization problem via improving edge weights on undirected graphs. Meyerson et al.~\cite{meyerson2009} proposed approximation algorithms for single-source and all-pair shortest paths minimization. Faster algorithms for some of these problems were also presented in~\cite{papagelis2011,parotisidis2015selecting}. Demaine et al.~\cite{demaine2010} minimized the diameter of a network by adding shortcut edges. Optimization of different node centralities by adding edges were studied in~\cite{crescenzi2015,ishakian2012framework,medya2018group}.


\section{Conclusions}

In this paper, we studied the influence minimization problem via edge deletion. Different from previous work, our formulation is data-driven, taking into account available propagation traces in the selection of edges. 
We have framed our problem under two different types of constraints---budget and matroid constraint. These variations were found to be APX-hard and cannot be approximated within a factor greater than $(1- \frac{1}{e})$. For the budget constrained version, we have developed an efficient greedy algorithm that achieves a good approximation guarantee by exploiting the monotonicity and submodularity of the objective function. The matroid constrained version was solved via continuous relaxation and a continuous greedy technique, achieving a probabilistic approximation guarantee. The experiments showed the effectiveness of our solutions, which outperform the baseline approaches, using both real and synthetic datasets.  


\section*{Appendix}

\subsection{Proof of Theorem \ref{thm:hardness}}

\begin{proof}
We prove the hardness result by reducing the known \textit{Influence Maximization} (IM) problem \cite{goyal2011} under CDM to BIL. Consider a problem instance $I_{IM}$ \cite{goyal2011}, where graph $G=(V,E)$, $|V|=n, |E|=m$ and integer $k$ are given.
We create a corresponding BIL problem instance ($I_{BIL}$) as follows. The directed social graph is $G'= (V',E')$ where $V'= V\cup \{x\}$, $x$ is an additional node. Let $C=\{(x,v) | v \in V\}$. In $I_{BIL}$, $E'= E \cup C$. We assume that the edges in $C$ are present for every action in IM. $C$ is also candidate set of edges. Let us assume the set $S$ (of size $k$) has the maximum influence ($\sigma^*$).  Now, it is easy to see that the maximum reduction of the influence of node $x$ in BIL can be obtained if and only if the edges ($k$ edges) between $x$ and $S$ are removed.
\end{proof}

\subsection{Proof of Lemma \ref{lemma:credit_computation2}}
\begin{proof}
If $w$ is not reachable from $v$, the proof becomes trivial. For $v\overrightarrow{a} w$  we use induction on length $l$. Let the set of reachable nodes via a path length of $l$ from $v$ in $G(a)$ be $R^a(v,l)$. We denote $N_{out}(u,a)=\{v|(u,v) \in E(a)\}$ and the decrease in credit contribution via the removal of the edge $e$ by any arbitrary node $w$ in $R^a(v,l)$ as $\delta^{l,w}_{a} (\{e\})$ and by all nodes in $R^a(v,l)$ as $\delta^{l}_{a} (\{e\})$. \\
\textbf{Base case:} when $l=0$, 
$\sum_{w \in V} \Gamma_{v,w}(a,0)= \Gamma_{v,v}=1$. So, the statement is true for $l=0$. \\
\textbf{Induction step:} Assume that the statement is true when restricted to path lengths $\ l$, for any arbitrary node $w$ where $w \in R^a(v,l)$, i.e.,  
     $\delta^{l,w}_{a} (\{e\}) = \big(\Gamma_{X,u}(a). \gamma_{(u,v)}(a)\big).  \Gamma_{v,w}(a,l) $

Notice that, $\delta^{l}_{a}(\{e\})= \big(\Gamma_{X,u}(a). \gamma_{(u,v)}(a)\big). \sum_{ w \in R^a(v,l)} \Gamma_{v,w}(a)=  \sum_{ w \in R^a(v,l)} \delta^{l,w}_{a} (\{e\})$. We will prove that the statement remains true for paths of length $l+1$ for nodes $w\in R^a(v,l+1)$. \\
Now in RHS,
\begin{equation*}
\begin{split}
  \sum_{w \in R^a(v,l+1)}\big(\Gamma_{X,u}(a). \gamma_{(u,v)}(a)\big).  \Gamma_{v,w}(a,l+1) \\
  = \sum_{w \in R^a(v,l+1)} \big(\Gamma_{X,u}(a). \gamma_{(u,v)}(a)\big). \sum_{y \in N_{in}(w)} \Gamma_{v,y}(a,l).\gamma_{(y,w)}(a)\\
  = \sum_{y \in R^a(v,l)}\big(\Gamma_{X,u}(a). \gamma_{(u,v)}(a)\big). \Gamma_{v,y}(a,l). \sum_{w \in N_{out}(y)} \gamma_{(y,w)}(a) \\
  =\sum_{y \in R^a(v,l)} \delta^{l,y}_{a} (\{e\}). \sum_{w \in N_{out}(y)} \gamma_{(y,w)}(a)\\
  = \sum_{w \in R^a(v,l+1)} \delta^{l+1,w}_{a} (\{e\})
    \end{split}
\end{equation*}
 \end{proof}

 \begin{algorithm}[t]
\caption{updateUC}
\begin{algorithmic}[1] 
 \REQUIRE $e=(u,v)$, $EP$, $UC$, $SC$
\FOR{$a \in \mathscr{A}$}
\STATE $\gamma\leftarrow EP[u][v][a] $
\FOR{each user $z$ such that $UC[z][u][a] >0$}
\FOR{each user $w$ such that $UC[v][w][a] >0$}
\STATE $UC[z][w][a]=UC[z][w][a]-(UC[z][u][a]\cdot\gamma)\cdot UC[v][w][a]$
\ENDFOR
\ENDFOR
\ENDFOR
\end{algorithmic}
\label{algo:updateUC}
\end{algorithm}

\begin{algorithm}[t]
\caption{updateSC}
\begin{algorithmic}[1] 
 \REQUIRE $e=(u,v)$, $EP$, $UC$, $SC$
\FOR{$a \in \mathscr{A}$ such that $SC[u][a]>0$ and $EP[u][v][a]>0$}
\STATE $\gamma\leftarrow EP[u][v][a] $
\FOR{each user $w$ such that $UC[v][w][a]>0$}
\STATE $SC[w][a]=SC[w][a]-(SC[u][a]\cdot \gamma)\cdot UC[v][w][a]$
\ENDFOR
\ENDFOR
\end{algorithmic}
\label{algo:updateSC}
\end{algorithm}

 \subsection{Algorithms \ref{algo:updateUC} and \ref{algo:updateSC} (updateUC and updateSC):}
 Method \textit{updateUC} (Algorithm~\ref{algo:updateUC}) identifies the credits (of the users) that has been changed upon an edge removal and does so by updating the data structure UC following the Observation \ref{obs:edge_removal2}. Method \textit{updateSC} do the same for the credits of target set of nodes (set $X$) by updating the data structure SC following the Observation \ref{obs:edge_removal3}.

\noindent 
 \subsection{Optimization of Greedy in BIL }
 We propose an intuitive and simple optimization technique to further improve the efficiency of Greedy. The question about optimization is the following: do all the edges in the candidate set ($C$) of edges need to be evaluated? To answer this, we introduce a concept of \textit{edge dominance}. The idea is very intuitive and simple. If an edge $e'=(w,x)$ is reachable from the target set through only a particular edge $e^*=(u,v)$ in all the DAGs, then we call $e'$ as the dominated and the edge $e^*$ as the dominating edge. In other words, there is no such path from a node in $X$ to $w$ without going through $e^*$ in $G(a)$ for all $a \in \mathscr{A}$. Note that the if the dominating edge $e^*$ is removed from the graph, the marginal contribution towards reducing influence of target set $X$ by removing $e'$ becomes $0$. The next lemma depicts the dominance of an edge.

\begin{lemma}
If $e'=(w,x)$ and $e* = (u,v)$ are present in $G(a)$, and $\Gamma_{X,w}=\Gamma_{X,u}\cdot \gamma_{(u,v)} \cdot \Gamma{v,w}$ for all $a \in \mathscr{A}$ then $e'$ is dominated by $e^*$.
\end{lemma}


\subsection{Proof of Theorem \ref{thm:continuous_submodular}}
\begin{proof}
Let $\mathcal{Y}=(y_1,y_2,...y_c)$ be the vector with membership probabilities for each edge in $C$ ($c=|C|$). Let the set $B$ be a random subset of $C$ where the edge $e_i\in C$ is included in set $B$ with probability $y_i$. If $f$ is the continuous extension of $\Delta$, then,$f(\mathcal{Y})= \textbf{E}_{B\sim \mathcal{Y}}[\Delta(B)] = \sum_{B\subseteq C} \Delta(B) \prod_{e_i\in B}{y_i}\prod_{e_i\in C\setminus B}{(1-y_i)}.$
To prove the function $f: [0,1]^C\rightarrow \mathbb{R}$ is a smooth monotone submodular function, we need to prove the followings:\\
i) $f$ has second partial derivatives everywhere. \\
ii) Monotonicity: For each $e_i\in C$, $\frac{\partial f}{\partial y_i}\geq 0 $. \\
iii) Submodularity: For each $e_i, e_j\in C$, $\frac{\partial^2 f}{\partial y_i \partial y_j}\geq 0$. \\
We derive a closed form similar in \cite{Vondrak08} for the second derivative and thus it always exists. 

For each $e_i\in C$, $\frac{\partial f}{\partial y_i}= \textbf{E}[\Delta(B)|e_i\in B]-\textbf{E}[\Delta(B)|e_i\notin B]$. As $\Delta$ is monotone, $\textbf{E}[\Delta(B)|e_i\in B]-\textbf{E}[\Delta(B)|e_i\notin B]\geq 0$ and thus, $f$ is also monotone.

For each $e_i, e_j\in C, i \neq j$, $\frac{\partial^2 f}{\partial y_i \partial y_j}= \textbf{E}[\Delta(B)|e_i, e_j\in B]-\textbf{E}[\Delta(B)|e_i \in B, e_j \notin B]-\textbf{E}[\Delta(B)|e_i\notin B, e_j\in B]-\textbf{E}[\Delta(B)|e_i, e_j \notin B]$. As $\Delta$ is submodular, $\frac{\partial^2 f}{\partial y_i \partial y_j}\geq 0$ from the above expression. Thus, $f$ is submodular. Note that if $i = j, \frac{\partial^2 f}{\partial y_i \partial y_j}= 0$. In other words, the relaxation $f$ is called multi-linear because it is linear in every co-ordinate ($y_i$).

\end{proof}


\subsection{APX-hardness of BIL} 
\begin{thm} \label{thm:apxhardness_BIL}
BIL is APX-hard and cannot be approximated within a factor greater than $(1-1/e)$.
\end{thm}

\begin{proof}
 We first reduce BIL from a similar problem as ILM that has matroid constraints with \textit{curvature with respect to optimal} as $1$. 
First we define a problem, ILM-O where maximum $b$ outgoing edges can be deleted form a node (unlike in ILM where the limit was on incoming edges). However, ILM-O is NP-hard, follows matroid constraints and has curvature $1$ (the proofs are straightforward and similar as in ILM) and thus cannot be approximated within a factor greater than $(1-\frac{1}{e})$ (similarly as Theorem \ref{thm:apxhardness_ILM}).  We give an $L$-reduction \cite{williamson2011design} from the ILM-O problem. The following two equations are satisfied in our reduction:
 \begin{equation*}
 \begin{split}
 OPT(I_{BIL}) \leq c_1\cdot OPT(I_{ILM-O}) \\
OPT(I_{ILM-O})-s(T^S) \leq c_2\cdot (OPT(I_{BIL})-s(T^B))
\end{split}
 \end{equation*}
where $I_{ILM-O}$ and $I_{BIL}$ are problem instances, and $OPT(Y)$ is the optimal value for instance $Y$. $s(T^S)$ and $s(T^B)$ denote any solution of the ILM-O and BIL instances, respectively. If the conditions hold and BIL has an $\alpha$ approximation, then ILM-O has an $(1-c_1c_2(1-\alpha))$ approximation. It is NP-hard to approximate ILM-O within a factor greater than $(1-\frac{1}{e})$. Now, $(1-c_1c_2(1-\alpha))\leq (1-\frac{1}{e})$, or, $\alpha \leq (1-\frac{1}{c_1 c_2e})$. So, if the conditions are satisfied, it is NP-hard to approximate BIL within a factor greater than $(1-\frac{1}{c_1 c_2e})$. 

Consider a problem instance $I_{ILM-O}$, where graph $G=(V,E)$, $|V|=n, |E|=m$ and integer $b$ and the target set $X=\{x\}$ are given. This problem becomes a BIL instance when $b=k$ where $k$ is the budget (in BIL). If the solution of $I_{ILM-O}$ is $s(T^S)$ then the influence of node $x$ will decrease by $s(T^S)$. Note that $s(T^B)= s(T^S)$ from the construction. Thus, both the conditions are satisfied when $c_1=1$ and $c_2=1$. So, BIL is NP-hard to approximate within a factor grater than $(1-\frac{1}{e})$.
\end{proof}

\subsection{Experimental Results for ILM}

 
 \begin{table}[h]
\centering

\begin{tabular}{| c | c | c | c | c|c|c|}
\hline
&\multicolumn{6}{c|}{\textbf{FXS:} $\#$ (tuples, actions)$\times 10^3$}  \\
\hline
\textbf{$\#$Edge Removed}& \multicolumn{2}{c|}{ ($30,1.7$)} & \multicolumn{2}{c|}{($50,4.8$)}& \multicolumn{2}{c|} {($75,6.9$)}\\
\hline
 & \textbf{CG}&\textbf{GRR} & \textbf{CG} & \textbf{GRR} & \textbf{CG} &  \textbf{GRR} \\
\hline
$20$ & $50$  & $44$ &  $48$ & $42$ & $51$ &  $44$ \\
\hline
$40$ &   $51$ & $47$& $53$ &   $45$ &$60$& $56$\\
\hline
$60$ &  $60$ &$ 54$ & $61$ &  $55$& $63$ & $57$\\
\hline
\end{tabular}

\caption{\textbf{[ILM] } Decrease in Influence ($\%$) in FXS by Continuous Greedy (CG) vs GRR varying the number of tuples. The number of tuples and actions are in thousands. \label{table:qual_scalability_tuple_ILM_FXS}}

 \end{table}

 \begin{table}[h]
\centering
\small
 
 \begin{tabular}{| c|c |c|c| c|c|}
\hline
\textbf{Dataset}& $|V|$  & $|C|$ & CG (Time) & CG (DI) & GRR (DI) \\
\hline
CY &  1.1m &   6.3k & 1858 & 55.1 & 47.2 \\
\hline
FX &  200k &  51k & 5690 & 46.2 & 37.3 \\
\hline
\end{tabular}
 
 \caption{\textbf{[ILM] } Running Time (Scalability) in seconds of CG and Decrease in Influence (percentage) by CG and GRR varying graph size for $|X|=20$, $b=2$ and the number of edges removed is $20$.  \label{table:scalability_graph_ILM}}

 \end{table}

 \textbf{Quality varying tuples:}
  Table \ref{table:qual_scalability_tuple_ILM_FCS} shows the results on FXS data. CG outperforms GRR by up to $8\%$. The results for FCS are in the main paper.
CG consistently produces better results than GRR.  
  
 \noindent
\textbf{Scalability varying graph size:} Table \ref{table:scalability_graph_ILM} shows the results varying the graph size. The running times are dominated by the size of the graphs and the candidate sets (the numbers of actions and tuples are same as in Table \ref{table:scalability_graph_BIL}). CG takes approximately 31 minutes on CY with 1m nodes and 6k candidate edges, where as it takes approximately 1.5 hours on FX with 200k nodes and 51k candidate edges. CG also outperforms GRR by up to $9\%$.

\textbf{Randomized Rounding vs Swap Rounding:}
We compare the results in terms of DI\% and running times taken by the rounding schemes (Randomized rounding (RR) and Swap rounding (SR) in Section \ref{sec:rounding}) on the solution set with edge probabilities generated by CG. 
Note that, RR is faster than SR as it is only a single-pass algorithm over the candidate set of edges. However, unlike SR, RR is not a loss-less scheme. Table \ref{table:qual_time_SR_RR_FXS} shows the results on FXS data where $b=1$. In practice, RR produces results of similar quality as in SR while being much faster.

\begin{table}[h]
\centering

\begin{tabular}{| c | c | c | c | c|c|c|}
\hline

\textbf{$\#$Edge Removed}& \multicolumn{2}{c|}{ Time} & \multicolumn{2}{c|}{DI\%}\\
\hline
 & \textbf{RR} & \textbf{SR} & \textbf{RR} & \textbf{SR} \\
\hline
$10$ & $746$ & $1300$  &  $31.7$ &  $33.2$  \\
\hline
$20$  &   $760$ & $1347$  &   $50.5$& $51.2$ \\
\hline
$30$  &  $769$  &$ 1288$ &   $58.6$ &$58.6$\\
\hline
$40$ &  $754$  &$ 1380$ &    $65$  &  $66.3$  \\
\hline
\end{tabular}

\caption{\textbf{[ILM] } Time (seconds) and Decrease in Influence (DI$\%$) in FXS by Continuous Greedy with Randomized Rounding (RR) vs Continuous Greedy with Swap Rounding (SR). The described time here is the total time taken by CG and the correspoding rounding procedure.  \label{table:qual_time_SR_RR_FXS}}
 \end{table}

\clearpage
\bibliographystyle{ACM-Reference-Format}
\bibliography{WWW19}

\end{document}